\documentclass[11pt]{article}
\pdfoutput=1 

\usepackage[utf8]{inputenc}
\usepackage{amsfonts,amssymb,amsmath,amsthm,complexity,physics,scalerel,float,dirtytalk,tabularx}
\usepackage[backref=page,colorlinks,bookmarks=true,citecolor=purple,linkcolor=purple,urlcolor=purple]{hyperref}
\usepackage[nameinlink,capitalize]{cleveref}
\usepackage[dvipsnames,table]{xcolor}
\usepackage{microtype,xspace,booktabs}
\usepackage{mathtools}
\usepackage[margin=1in]{geometry}
\usepackage{enumerate}
\usepackage{bbm}
\usepackage{footnote}
\usepackage{enumitem}
\usepackage{soul}
\usepackage{graphicx}
\usepackage{tikz}
\usepackage[labelfont=bf]{caption}
\usepackage{thmtools}
\usepackage{xfrac}

\Crefname{section}{Sec.}{Sections}
\crefname{subsection}{Sec.}{Sections}
\crefname{prop}{Proposition}{Propositions}

\renewcommand{\backref}[1]{}

\renewcommand{\backrefalt}[4]{
\ifcase #1 
\or
[p.\ #2]
\else
[pp.\ #2]
\fi}

\renewcommand{\hat}{\widehat}

\usetikzlibrary{arrows.meta,positioning} 

\makeatletter
\def\@cite#1#2{\textup{[{#1\if@tempswa , #2\fi}]}}
\makeatother

\renewcommand{\abs}[1]{\left\lvert#1\right\rvert}
\renewcommand{\norm}[1]{\left\lVert#1\right\rVert}
\DeclarePairedDelimiterX{\inp}[2]{\langle}{\rangle}{#1, #2}
\DeclarePairedDelimiter\ceil{\lceil}{\rceil}
\DeclarePairedDelimiter\floor{\lfloor}{\rfloor}

\declaretheorem[name=Theorem]{theorem}
\declaretheorem[name=Lemma, sibling=theorem]{lemma}
\declaretheorem[name=Corollary,sibling=theorem]{corollary}
\declaretheorem[name=Proposition,sibling=theorem]{prop}
\declaretheorem[name=Claim, sibling=theorem]{claim}
\declaretheorem[name=Fact,sibling=claim]{fact}
\declaretheorem[name=Question, style=definition]{question}

\declaretheorem[name=Definition, sibling=theorem, style=definition]{definition}
\declaretheorem[name=Problem, sibling=theorem]{problem}
\declaretheorem[name=Conjecture]{conjecture}

\renewcommand{\deg}{\mathrm{deg}}
\newcommand{\dt}{\mathrm{D}}
\newcommand{\calD}{\mathcal{D}}

\newcommand{\cert}{\mathrm{C}}
\DeclareMathOperator{\rdeg}{rdeg}
\newcommand{\rdegnull}{\mathrm{rdeg}}
\newcommand{\rdegeps}{\mathrm{\rdeg}_{\epsilon}}
\DeclareMathOperator{\ndeg}{ndeg}

\newcommand{\sgn}{\mathrm{sgn}}
\newcommand{\sens}{\mathrm{s}}
\newcommand{\senszero}{\mathrm{s}^{(0)}}
\newcommand{\sensone}{\mathrm{s}^{(1)}}
\newcommand{\bs}{\mathrm{bs}}

\newcommand{\conj}{\overline}
\newcommand{\quantumquery}{\func{Q}}
\newcommand{\postquantumquery}{\func{PostQ}}

\newcommand{\eps}{\varepsilon}
\newcommand{\pmone}{\{-1,1\}}
\newcommand{\bool}{\{0,1\}}
\newcommand{\lpar}{\left(}
\newcommand{\rpar}{\right)}
\newcommand{\Thr}{\mathrm{Thr}}
\renewcommand{\R}{\mathbb{R}}
\newcommand{\anddim}{\dim_{\wedge}}
\newcommand{\ordim}{\dim_{\vee}}
\newcommand{\tildex}{\tilde{x}}
\newcommand{\tildey}{\tilde{y}}
\DeclareMathOperator{\MT}{\mathsf{MT}}

\definecolor{conj}{HTML}{C2C0BF} 
\definecolor{open}{HTML}{A31F34} 

\newcommand{\ct}[2]{%
\cellcolor{white}%
\begin{tabular}[t]{@{}c@{}}#1\\[-5pt]#2\end{tabular}%
}
\newcolumntype{Y}{>{\centering\arraybackslash}X}

\renewcommand{\co}[2]{%
\cellcolor{conj!30}%
\begin{tabular}[t]{@{}c@{}}#1\\[-5pt]#2\end{tabular}%
}

\renewcommand{\cc}[2]{%
\cellcolor{conj!30}%
\begin{tabular}[t]{@{}c@{}}#1\\[-5pt]#2\end{tabular}%
}

\newcommand{\smcite}[1]{{\scriptsize \cite{#1}}}

\newcommand{\detq}{\mathrm{D}}
\newcommand{\randq}{\mathrm{R}}
\newcommand{\randcert}{\mathrm{RC}}
\newcommand{\quantq}{\mathrm{Q}}
\newcommand{\s}{\mathrm{s}}

\DeclareMathOperator{\AND}{\mathsf{AND}}
\DeclareMathOperator{\NOTgate}{\mathsf{NOT}}
\DeclareMathOperator{\XOR}{\mathsf{XOR}}
\DeclareMathOperator{\OR}{\mathsf{OR}}
\DeclareMathOperator{\PARITY}{\mathsf{PARITY}}

\newcommand{\apxdeg}{\adeg}
\newcommand{\adeg}{\widetilde{\vphantom{t}\smash{\deg}}}

\newclass{\ptime}{P}
\newclass{\bqp}{BQP}
\newclass{\eqp}{EQP}
\newclass{\pp}{PP}
\newclass{\rp}{RP}
\newclass{\bpp}{BPP}
\newclass{\posteqp}{PostEQP}
\newclass{\postzpp}{PostZPP}
\newclass{\postbqp}{PostBQP}
\newclass{\ceqp}{\CeP}
\newclass{\coeqp}{\mathsf{co}\CeP}

\newclass{\cA}{A}
\newclass{\cB}{B}
\newclass{\epp}{EPP}

\newcommand{\newprob}[2]{\newcommand{#1}{{\text{\upshape\scshape #2}}\xspace}}
\newprob{\pA}{A}
\newprob{\pB}{B}
\newprob{\orprob}{OR}
\newprob{\maj}{MAJ}
\newprob{\majn}{MajOrNone}
\newprob{\imbalance}{Imbalance}
\newprob{\parprob}{Parity}
\newprob{\bi}{BI}
\renewcommand{\epsilon}{\varepsilon}

\newcommand{\dom}{\mathrm{Dom}}

\newcommand\bigO{
  \mathcal{O}
  }
\newcommand\littleO{o}
\newcommand{\newbound}[1]{\colorbox{green!50}{#1}}

\title{On the Rational Degree of Boolean Functions and Applications}

\date{\vspace{-3ex}}

\begin{document}

\maketitle

    \vspace{-8ex}

\begin{tabular}{cccc}
Vishnu Iyer&
Siddhartha Jain&
Robin Kothari&
Matt Kovacs-Deak\\[-1mm]
\small\slshape UT Austin&
\small\slshape UT Austin&
\small\slshape Google Quantum AI&
\small\slshape University of Maryland\\
\\
Vinayak M. Kumar&
Luke Schaeffer&
Daochen Wang&
Michael Whitmeyer\\[-1mm]
\small\slshape UT Austin&
\small\slshape University of Waterloo&
\small\slshape UBC&
\small\slshape University of Washington
\end{tabular}

\begin{abstract}
We study a natural complexity measure of Boolean functions known as the rational degree. Denoted $\rdeg(f)$, it is the minimal degree of a rational function that is equal to $f$ on the Boolean hypercube.
For total functions $f$, it is conjectured that $\rdeg(f)$ is polynomially related to the Fourier degree of $f$, $\deg(f)$.
Towards this conjecture, we show that:
\begin{itemize}[label=\small$\bullet$,leftmargin=1.5em]
    \item Symmetric functions have rational degree at least $\Omega(\deg(f))$ and unate functions have rational degree at least $\sqrt{\deg(f)}$. We observe that both of these lower bounds are asymptotically tight. 
    \item Read-once $\AC$ and $\TC$ formulae have rational degree at least $\Omega(\sqrt{\deg(f)})$. If these formulae contain parity gates, we show a lower bound of $\Omega(\deg(f)^{1/2d})$, where $d$ is the depth. 
    \item Almost every Boolean function on $n$ variables has rational degree at least $n/2 - \bigO(\sqrt{n})$.
\end{itemize}
In contrast, we exhibit partial functions that witness unbounded separations between rational and approximate degree, in both directions. As a consequence, we show that for quantum computers, post-selection and bounded-error are incomparable resources in the black-box model.

In addition, we show $\AND$ and $\OR$ composition lemmas for the rational degree and exhibit new polynomial separations between the rational degree and other well-studied complexity measures, such as sensitivity and spectral sensitivity.
\end{abstract}

\section{Introduction}

Starting with the seminal work of Minsky and Papert \cite{minskypapert}, a long line of research has sought to relate various measures of Boolean function complexity. In \cite{nisan-szegedy}, Nisan and Szegedy proved that the deterministic decision tree complexity $\dt(f)$ of a Boolean function $f$ is polynomially related to its degree $\deg(f)$. The same paper posed two open questions. One of them conjectures that the sensitivity and block sensitivity of a Boolean function are polynomially related. This conjecture was recently proven in a breakthrough by Huang \cite{huang2019induced}. Huang's result brought sensitivity into a \say{happy flock} of complexity measures on total Boolean functions that are all polynomially related: sensitivity, degree, approximate degree, and notions of query complexity.

Another natural measure of Boolean function complexity is the minimal degree of the ratio of two polynomials which represents the function exactly, called the \emph{rational degree} (denoted $\rdeg$). However, $\rdeg$ is \emph{not} known to be either polynomially related to or separated from the complexity measures mentioned above.
In fact, this was the other open question posed over 30 years ago in the paper of Nisan and Szegedy (via personal communication with Fortnow) \cite{nisan-szegedy}. This question was reiterated by Aaronson \textit{et al.} \cite{aaronson2020degree} yet very little progress has been made toward its resolution.

\begin{question}[Fortnow \cite{nisan-szegedy}]
\label{q:fortnow}
    Does there exist $c > 1$ such that for all total Boolean functions $f$, $\deg(f) \leq \bigO(\rdeg(f)^c)$?
\end{question}

One of the motivations for Fortnow's question was complexity-theoretic: is the intersection of \ceqp\ and \coeqp\ strictly contained in \pp\ with respect to an oracle \cite{FortnowBlog}? \ceqp\ and \coeqp\ are \say{counting classes} \cite{ComplexityZoo} which we define later, and rational degree corresponds to the black-box version of their intersection.

The rational degree is also related to quantum query complexity. In particular, de Wolf defined the \emph{non-deterministic degree} $\ndeg(f)$ of a Boolean function $f$ as the minimal degree of a polynomial whose zero set is precisely the set of inputs on which $f$ evaluates to false \cite{dewolf2000}, and related it to the rational degree through the identity $\rdeg(f) = \max\{\ndeg(f), \ndeg(\bar{f})\}$. de Wolf also proved that the non-deterministic degree $\ndeg(f)$ equals the \emph{non-deterministic quantum query complexity} up to a constant factor. 

In the same manuscript, de Wolf stated the following conjecture which, together with the inequality $\deg(f) \leq \dt(f)$, would resolve Fortnow's question in the affirmative with $c = 2$.
\begin{conjecture}[de Wolf \cite{dewolf2000}]
    For all Boolean functions $f$, $\dt(f) \leq \bigO(\ndeg(f) \ndeg(\bar{f}))$.
\end{conjecture}

Mahadev and de Wolf showed \cite{mahadev2014} an even tighter connection between the notion of rational degree and quantum query complexity: denoting by $\rdegeps(f)$ the minimum degree of a rational polynomial that $\epsilon$-approximates $f$ pointwise, they showed that $\rdegeps(f)$ equals (up to a constant factor) the query complexity of an $\eps$-error quantum algorithm that computes $f$ with \emph{post-selection}, an operation that allows for projection onto an efficiently computable set of basis states. For partial functions, it can be shown that the rational degree gives a lower bound on the query complexity of algorithms with post-selection, though the opposite direction is not known to be true. Furthermore, this result extends to the case of $\eps = 0$, the so-called \say{zero-error} setting.

\subsection{Our Results}

We prove lower bounds on the rational degree for certain classes of total Boolean functions. We summarize our results, which are also depicted in \Cref{fig:results-table}:
\begin{itemize}[leftmargin=0.5in]
    \item[\Cref{subsec:symmetric}] For symmetric functions we show that $\rdeg(f) \geq (\deg(f) + 1)/3$. This lower bound is tight up to a factor of $3/2$, as witnessed by the $\PARITY_n$ function. Our technique generalizes to give a $\deg(f)/4$ lower bound for (possibly partial) functions which are constant on many Hamming slices. 
    \item[\Cref{subsec:monotone}] For unate functions we prove that $\rdeg(f) = \sens(f) \geq \sqrt{\deg(f)}$. This is also tight as witnessed by the $\AND_n \circ \OR_n$ function.
    \item[\Cref{subsec:readonce}] We employ the lower bound on symmetric and unate functions to show that for $f$ computable by read-once depth-$d$ $\AC$ and $\TC$ formulae,\footnote{Read-once $\AC$ (resp. $\TC$) are read-once formulae consisting of unbounded fan-in $\{\lor,\land,\neg\}$ gates (resp. arbitrary linear threshold gates). $\AC[\oplus]$ and $\TC[\oplus]$ refer to the classes with the addition of unbounded fan-in parity gates.} $\rdeg(f) \geq \Omega(\sqrt{\deg(f)})$. This is tight, as witnessed by the $\AND_n \circ \OR_n$ function. For read-once $\TC$ formulae, we also give a standalone proof for a somewhat strengthened result in \cref{app:readonce_tc}. If we include parity gates, the lower bounds become $\Omega(\deg(f)^{1/d})$ for read-once $\AC[\oplus]$ and $\Omega(\deg(f)^{1/2d})$ for read-once $\TC[\oplus]$.
    \item[\Cref{subsec:random}] Our final lower bound on total functions is extremal: we prove that almost all Boolean functions on $n$ bits have rational degree at least $n/2 - \bigO(\sqrt{n})$.
\end{itemize}

\begin{figure}
\centering
\label{fig:results-table}
{\rowcolors{2}{white}{black!05}
\renewcommand{\arraystretch}{1.5}
\begin{tabularx}{\textwidth}{l Y Y}
    & \textbf{Lower Bound} & \textbf{Upper Bound} \\
    Symmetric Functions & $(\deg(f) + 1)/3$ & $\deg(f)/2$\quad$(\PARITY_n)$ \\
    Unate Functions & $\sqrt{\deg(f)}$ & $\sqrt{\deg(f)}$\quad$(\AND_{n} \circ \OR_{n})$ \\
    Read-once $\AC$/$\TC$  & $\Omega(\sqrt{\deg(f)})$ & $\sqrt{\deg(f)}$\quad$(\AND_{n} \circ \OR_{n})$ \\ 
    Read-once depth $d$ $\AC[\oplus]$ & $\Omega(\deg(f)^{1/d})$ & $\sqrt{\deg(f)}$\quad$(\AND_{n} \circ \OR_{n})$ \\ 
    Read-once depth $d$ $\TC[\oplus]$ & $\Omega(\deg(f)^{1/2d})$  & $\sqrt{\deg(f)}$\quad$(\AND_{n} \circ \OR_{n})$ \\ 
    Almost all $f\colon \bool^n \to \bool$ & $n/2 - \bigO(\sqrt{n})$ & $n$\\
\end{tabularx}
}
\caption{A table summarizing our lower bounds on rational degree for total functions. The third column gives an example of a function that demonstrates the asymptotic tightness of our lower bound, where applicable.}
\end{figure}

Next, we prove that the non-deterministic degree composes exactly with respect to $\AND$. Leveraging this result, we:
\begin{itemize}[leftmargin=0.5in]
    \item[\Cref{sec:composition-lemmas}] Establish $\AND$ and $\OR$ composition lemma for the rational degree.
    \item[\Cref{sec:rdeg-v-lambda}] Exhibit an asymptotic separation between the rational degree and the spectral sensitivity of degree $\sfrac{3}{2}$.
\end{itemize}

In contrast to total functions, we show that for partial functions, the rational and approximate degrees can be unboundedly separated in both directions. These separations also resolve an open question of Fortnow \cite{FortnowBlog}. 

\vspace{3mm}

\begin{itemize}[leftmargin=0.5in]
    \item[\Cref{sec:rdeg-is-high}]
    We give a partial function $\majn_n$ on $n$ bits with constant quantum query complexity yet rational degree $\Omega(n)$. As a result, $\majn_n$  has constant approximate degree and $\Omega(n)$ zero-error post-selected quantum query complexity.
    
    \item[\Cref{sec:rdeg-is-low}] On the other hand, we give a partial function $\imbalance_n$ on $n$ bits with approximate degree $\Omega(n)$ yet constant rational degree. As a result, $\imbalance_n$  has constant zero-error post-selected quantum query complexity and quantum query complexity $\Omega(n)$.
\end{itemize}

Now, employing the framework of standard complexity results such as \cite{FSS1981}, we can argue that post-selection and bounded error are incomparable resources in the black-box setting. To formalise this, we define \posteqp\ as the class of decision problems which can be decided deterministically in polynomial time by quantum computers with access to post-selection. In particular, there exist bidirectional oracle separations between $\posteqp$ and $\bqp$. Formally, combining the results of \Cref{cor:posteqp-high,cor:posteqp-low} we get the following statement.

\begin{corollary}
    There exist oracles $O_1$ and $O_2$ such that $\bqp^{O_1} \not \subseteq \posteqp^{O_1}$ yet $\posteqp^{O_2} \not \subseteq \bqp^{O_2}$.
\end{corollary}

These complexity-theoretic consequences are summarized in \Cref{fig:quantum_hierarchy}. As the figure illustrates, these are the strongest possible separations in the black-box model.

\begin{figure}[t!]
    \centering
    \begin{tikzpicture}[scale=1.1]
\tikzset{inner sep=0,outer sep=3}

\tikzstyle{a}=[inner sep=6pt, inner ysep=6pt,outer sep=0.5pt,
draw=black!20!white, fill=Cerulean!2!white, very thick, rounded corners=6pt, align=center]

\begin{scope}[yscale=1.145]
\large
\node[a] (P) at (0,1.5) {$\ptime$};
\node[a] (RP) at (-3,3) {$\rp$};
\node[a] (BPP) at (-3,5) {$\bpp$};
\node[a] (EQP) at (0,3) {$\eqp$};
\node[a] (BQP) at (0,6) {$\bqp$};
\node[a] (PostEQP) at (3,4) {$\posteqp$};
\node[a] (PostBQP) at (3,7) {$\postbqp$};
\end{scope}

\path[-{Stealth[length=6pt]},line width=.6pt,gray]
(P) edge (RP)
(P) edge (EQP)
(P)  edge (PostEQP)
(RP) edge (BPP)
(EQP) edge (BQP)
(EQP) edge (PostEQP)
(BQP) edge (PostBQP)
(BPP) edge (BQP)
(PostEQP) edge (PostBQP);

\small
\hypersetup{hidelinks}
\tikzset{new/.style={-{Stealth[length=6pt]},dashed,line width=1pt,purple!80}}
\draw[new,bend left=7]
(RP) edge
node[midway,above,inner sep=4pt,sloped,pos=0.35] {\Cref{cor:posteqp-low}}
(PostEQP);
\draw[new,bend left=20]
(PostEQP) edge
node[pos=0.5,above,sloped,inner sep=3pt] {\Cref{cor:posteqp-high}}
(BQP);

\end{tikzpicture}
    \caption{Relevant complexity classes. We are able to obtain the strongest possible oracle separations in this picture. An arrow $\cA\rightarrow\cB$ means $\cA\subseteq \cB$ relative to all oracles. A~dashed arrow~$\cA\dashrightarrow\cB$ means $\cA\not\subseteq \cB$ relative to some oracle.}
    \label{fig:quantum_hierarchy}
\end{figure}

In addition to these consequences for $\posteqp$, our lower bound also resolves Fortnow's complexity-theoretic question. We show that not only is $\ceqp \cap \coeqp$ strictly contained in \pp, but even \rp\ is not in this intersection with respect to an oracle. The class \ceqp\ is the set of languages decidable by an \NP\ machine such that if the string is in the language, the number of accepting paths is \emph{exactly} equal to the number of rejecting paths. Finally, to contextualize the power of $\posteqp$, we provide strong evidence that zero-error post-selection can offer advantage over efficient classical computation, even in the non-relativized setting.

\begin{itemize}[leftmargin=0.5in]
    \item[\Cref{sec:post-selection-nondeterminism}] We show that $\posteqp$ contains $\NP \cap \coNP$. We remark that $\NP \cap \coNP$ is not even believed to be contained in $\bpp$.
\end{itemize}

\setlength{\tabcolsep}{2pt}
\renewcommand{\arraystretch}{1.3}
\begin{table}
\hspace{-3em}
\begin{minipage}{\linewidth}
\begin{center}
\caption{Best known separations between complexity measures for total functions (Reproduced from \cite[Table 1]{aaronson2020degree} and updated with recent developments as well as our results)}
\begin{tabular}{r||c|c|c|c|c|c|c|c|c|c|c|c|c}
\label{tab:relations}
{} & $\detq$ & $\randq_0$ & $\randq$ & $\cert$ & $\randcert$ & $\bs$ & $\s$ & $\lambda$ & $\quantq_E$ & $\deg$ & $\quantq$ & $\adeg$ & $\rdeg$ \\ \hline\hline

$\detq$ &
\cellcolor{darkgray} & 
\ct{2, 2}{\smcite{ABB+15}} & 
\cc{2, 3}{\smcite{ABB+15}} &
\ct{2, 2}{$\wedge\circ\vee$} & 
\cc{2, 3}{$\wedge\circ\vee$} & 
\cc{2, 3}{$\wedge\circ\vee$} &
\co{3, 6}{\smcite{BHT17}} &
\co{4, 6}{\smcite{ABB+15}} &
\co{2, 3}{\smcite{ABB+15}} & 
\co{2, 3}{\smcite{GPW15}} & 
\ct{4, 4}{\smcite{ABB+15}} &
\ct{4, 4}{\smcite{ABB+15}} &
\co{\newbound{2}, ?}{$\wedge\circ\overline{\EH}$}\\ \hline

$\randq_0$ &
\ct{1, 1}{$\oplus$} &
\cellcolor{darkgray} & 
\ct{2, 2}{\smcite{ABB+15}} &
\ct{2, 2}{$\wedge\circ\vee$} & 
\cc{2, 3}{$\wedge\circ\vee$} & 
\cc{2, 3}{$\wedge\circ\vee$} &
\co{3, 6}{\smcite{BHT17}} &
\co{4, 6}{\smcite{ABB+15}} &
\co{2, 3}{\smcite{ABB+15}} &
\co{2, 3}{\smcite{GJPW15}} &
\co{3, 4}{\smcite{ABB+15}} & 
\ct{4, 4}{\smcite{ABB+15}} &
\co{\newbound{2}, ?}{$\wedge\circ\overline{\EH}$} \\ \hline

$\randq$ &
\ct{1, 1}{$\oplus$} & 
\ct{1, 1}{$\oplus$} & 
\cellcolor{darkgray} & 
\ct{2, 2}{$\wedge\circ\vee$} &
\cc{2, 3}{$\wedge\circ\vee$} & 
\cc{2, 3}{$\wedge\circ\vee$} &
\co{3, 6}{\smcite{BHT17}} &
\co{4, 6}{\smcite{ABB+15}} &
\co{$\frac{3}{2}$, 3}{\smcite{ABB+15}} &
\co{2, 3}{\smcite{GJPW15}} &
\co{3, 4\vspace{-0.75ex}}{\smcite{BS20}\\[-2ex]\smcite{SSW20}\vspace{-0.4ex}} & 
\ct{4, 4}{\smcite{ABB+15}} &
\co{\newbound{2}, ?}{$\wedge\circ\overline{\EH}$}\\ \hline

$\cert$ &
\ct{1, 1}{$\oplus$} & 
\ct{1, 1}{$\oplus$} & 
\co{1, 2}{$\oplus$} & 
\cellcolor{darkgray} &
\ct{2, 2}{\smcite{GSS13}} & 
\ct{2, 2}{\smcite{GSS13}} &
\co{3, 5}{\smcite{BBGJK21}} &
\co{4, 6}{\smcite{BBGJK21}} &
\co{1.15, 3}{\smcite{Amb13}}  & 
\co{2, 3}{\smcite{BBGJK21}} & 
\co{2, 4}{$\wedge$} & 
\co{3, 4}{\smcite{BBGJK21}} &
\co{\newbound{2}, ?}{$\wedge\circ\overline{\EH}$}\\ \hline

$\randcert$ &
\ct{1, 1}{$\oplus$} & 
\ct{1, 1}{$\oplus$} & 
\ct{1, 1}{$\oplus$} & 
\ct{1, 1}{$\oplus$} & 
\cellcolor{darkgray} & 
\co{$\frac{3}{2}$, 2}{\smcite{GSS13}} &
\co{2, 4}{\smcite{Rub95}} &
\co{2, 4}{$\wedge$} &
\co{1.15, 2}{\smcite{Amb13}} & 
\co{1.63, 2}{\smcite{nisanwigderson1995}} & 
\ct{2, 2}{$\wedge$} & 
\ct{2, 2}{$\wedge$} &
\co{\newbound{2}, ?}{$\wedge\circ\overline{\EH}$}\\ \hline

$\bs$ &
\ct{1, 1}{$\oplus$} & 
\ct{1, 1}{$\oplus$} & 
\ct{1, 1}{$\oplus$} & 
\ct{1, 1}{$\oplus$} & 
\ct{1, 1}{$\oplus$} & 
\cellcolor{darkgray} &
\co{2, 4}{\smcite{Rub95}} &
\co{2, 4}{$\wedge$} &
\co{1.15, 2}{\smcite{Amb13}} & 
\co{1.63, 2}{\smcite{nisanwigderson1995}} & 
\ct{2, 2}{$\wedge$} &
\ct{2, 2}{$\wedge$} &
\co{\newbound{2}, ?}{$\wedge\circ\overline{\EH}$}\\ \hline 

$\s$ &
\ct{1, 1}{$\oplus$} &
\ct{1, 1}{$\oplus$} &
\ct{1, 1}{$\oplus$} &
\ct{1, 1}{$\oplus$} &
\ct{1, 1}{$\oplus$} &
\ct{1, 1}{$\oplus$} &
\cellcolor{darkgray} &
\ct{2, 2}{$\wedge$} &
\co{1.15, 2}{\smcite{Amb13}} &
\co{1.63, 2}{\smcite{nisanwigderson1995}}  &
\ct{2, 2}{$\wedge$} & 
\ct{2, 2}{$\wedge$} &
\co{\newbound{2}, ?}{$\wedge\circ\overline{\EH}$}\\ \hline

$\lambda$ &
\ct{1, 1}{$\oplus$} &
\ct{1, 1}{$\oplus$} &
\ct{1, 1}{$\oplus$} &
\ct{1, 1}{$\oplus$} &
\ct{1, 1}{$\oplus$} &
\ct{1, 1}{$\oplus$} &
\ct{1, 1}{$\oplus$} &
\cellcolor{darkgray} &
\ct{1, 1}{$\oplus$} &
\ct{1, 1}{$\oplus$} &
\ct{1, 1}{$\oplus$} &
\ct{1, 1}{$\oplus$}  &
\co{\newbound{1.5}, ?}{$\wedge\circ\overline{\EH}$}\\ \hline

$\quantq_E$ &
\ct{1, 1}{$\oplus$} &
\co{1.33, 2}{$\bar{\wedge}$-tree} &
\co{1.33, 3}{$\bar{\wedge}$-tree} &
\ct{2, 2}{$\wedge\circ\vee$} &
\cc{2, 3}{$\wedge\circ\vee$} &
\cc{2, 3}{$\wedge\circ\vee$} &
\co{3, 6}{\smcite{BHT17}} &
\co{4, 6}{\smcite{aaronson2016separations}} &
\cellcolor{darkgray} &
\co{2, 3}{\smcite{aaronson2016separations}} &
\co{2, 4}{$\wedge$} &
\ct{4, 4}{\smcite{aaronson2016separations}} &
\co{\newbound{2}, ?}{$\wedge\circ\overline{\EH}$}\\ \hline

$\deg$ & 
\ct{1, 1}{$\oplus$} &
\co{1.33, 2}{$\bar{\wedge}$-tree} &
\co{1.33, 2}{$\bar{\wedge}$-tree} &
\ct{2, 2}{$\wedge\circ\vee$} &
\ct{2, 2}{$\wedge\circ\vee$} &
\ct{2, 2}{$\wedge\circ\vee$} &
\ct{2, 2}{$\wedge\circ\vee$} &
\ct{2, 2}{$\wedge$} &
\ct{1, 1}{$\oplus$}&
\cellcolor{darkgray} &
\ct{2, 2}{$\wedge$} &
\ct{2, 2}{$\wedge$} &
\co{\newbound{2}, ?}{$\wedge\circ\overline{\EH}$}\\ \hline

$\quantq$ &
\ct{1, 1}{$\oplus$} &
\ct{1, 1}{$\oplus$} &
\ct{1, 1}{$\oplus$} &
\ct{2, 2}{\smcite{aaronson2016separations}} &
\cc{2, 3}{\smcite{aaronson2016separations}} &
\cc{2, 3}{\smcite{aaronson2016separations}} &
\co{3, 6}{\smcite{BHT17}} &
\co{4, 6}{\smcite{aaronson2016separations}} &
\ct{1, 1}{$\oplus$} &
\co{2, 3}{\smcite{aaronson2016separations}} &
\cellcolor{darkgray} &
\ct{4, 4}{\smcite{aaronson2016separations}} &
\co{\newbound{1.5}, ?}{$\wedge\circ\overline{\EH}$}\\ \hline

\raisebox{-2pt}{$\adeg$} &
\ct{1, 1}{$\oplus$} &
\ct{1, 1}{$\oplus$} &
\ct{1, 1}{$\oplus$} &
\ct{2, 2}{\smcite{bun17optimal}} &
\ct{2, 2}{\smcite{bun17optimal}} &
\ct{2, 2}{\smcite{bun17optimal}} &
\ct{2, 2}{\smcite{bun17optimal}} &
\ct{2, 2}{\smcite{bun17optimal}} &
\ct{1, 1}{$\oplus$} &
\ct{1, 1}{$\oplus$}  &
\ct{1, 1}{$\oplus$} &
\cellcolor{darkgray} &
\co{\newbound{1.5}, ?}{$\wedge\circ\overline{\EH}$}\\ \hline

\end{tabular}
\label{tab:sep}
\end{center}
\begin{itemize}[itemsep=3pt,label=\small$\bullet$]
    \item  An entry $a,b$ in the row $M_1$ and column $M_2$ roughly means that there exists a function $g$ with $M_1(g) \geq M_2(g)^{a-o(1)}$, and for all total functions $f$, $M_1(f) \leq M_2(f)^{b+o(1)}$. For example, the $3,4$ entry at row $\randq$ and column $\quantq$ means that the maximum possible separation between $\randq$ and $\quantq$ is at least cubic and at most quartic.
    \item {The second row of each cell contains an example of a function that achieves the separation (or a citation to an example), 
    where $\oplus = \textsc{parity}$, 
    $\wedge = \textsc{and}$, 
    $\vee = \textsc{or}$, 
    $\wedge \circ \vee = \textsc{and-or}$, $\bar{\wedge}$-tree is the balanced \textsc{nand}-tree function, and $\overline{\EH}$ is the negation of the exact-half function defined in \cref{sec:rdeg-v-lambda}.}
    \item Cells have a white background if the relationship is optimal and a gray background otherwise.
    \item The entries with \newbound{green} background in the column for $\rdeg$ are contributions of this work (see \cref{prop:separation}). Getting any polynomial upper bound (replacing \say{?} on the right with a constant) in any cell would resolve \cref{q:fortnow}.
\end{itemize}
\end{minipage}
\end{table}
\renewcommand{\arraystretch}{1}
\cref{tab:relations} is a reproduction of \cite[Table 1]{aaronson2020degree} showing the  known relationships between various Boolean complexity measures. We extended this table with a column corresponding to the rational degree that shows the separations established in our work, and updated some of the other cells to reflect more recent developments.

\section{Preliminaries}

    In this section we review some of the notation and definitions used in our paper. We denote by $[n]$ the set $\{1,2,...,n\}$. Given a function $f\colon S \to \R$ we denote by $\norm{f}_1$ its $l_1$ norm, $\norm{f}_1 = \sum_{x\in S} \abs{f(x)}$. For a bitstring $x \in \bool^n$, we denote by $|x|$ the Hamming weight of $x$: the number of indices equal to $1$. If $x \in \pmone^n$ the Hamming weight is the number of bits that equal $-1$. For a string $x \in \bool^n$, $x^{i \to b}$ refers to the string obtained by replacing the $i$th bit of $x$ with $b$.

\subsection{Boolean Functions}
    A (total) Boolean function is any function $f\colon \Sigma^n \to \Sigma$ where $\Sigma$ is some two-element set. We will refer to the set $\Sigma^n$ as the Boolean hypercube. We will primarily work over the sets $\Sigma = \bool$ and $\Sigma = \pmone$, often swapping between them to make our analysis more presentable.
    While not all Boolean complexity measures are left invariant by this change of representation, all of the measures considered in this paper are preserved.
    
    We also consider restrictions of Boolean functions to proper subsets of the Boolean cube $D \subset \Sigma^n$. We refer to such functions $f\colon D \to \Sigma$ as \emph{partial} Boolean functions. Given a Boolean function $f$, we denote its negation by $\bar{f}$.

    We can define an inner product on the space of functions $f\colon \pmone^n \to \mathbb{R}$:
    \begin{align*}
        \inp{f}{g} = 2^{-n} \sum_{x\in \{-1,1\}^n} f(x)g(x).
    \end{align*}
    For each $S \subseteq [n]$ we define the \emph{character function $\chi_S$ on $S$} as $\chi_S(x) = \prod_{i\in S} x_i$. The character functions $\chi_S$ form an orthonormal basis under the above inner product. Thus each function over $\pmone^n$ can be uniquely expressed via its \emph{Fourier representation}: 
    \begin{align*}
        f = \sum_{S\subseteq[n]} \hat{f}(S)\cdot \chi_S,
    \end{align*}
    where we refer to $\hat{f}(S) = \inp{f}{\chi_S}$ as the \emph{Fourier coefficient of $f$ at $S$}. We say an input $i \in [n]$ is \emph{relevant} for $f$ if $x_i$ appears in the Fourier expansion for $f$. In other words, $f$ depends on $x_i$ in a nontrivial manner.

    In this work, we consider Boolean functions with special properties. For example, we say a function $f$ is \emph{symmetric} if it remains invariant under any permutation of the input variables. A Boolean function $f\colon \bool^n \to \bool$ is said to be \emph{monotone} if $\forall x,y \in \bool^n$, $x\leq y$ implies $f(x) \leq f(y)$, where $x \leq y$ is taken pointwise. Finally, the notion of unateness generalizes monotonicity. We say a Boolean function $f\colon \bool^n \to \bool$ is \emph{unate} in the coordinate $i$ if either:
    \begin{enumerate}[label=(\arabic*)]
        \item For every $x \in \bool^n$, $f(x^{i \to 0}) \leq f(x^{i \to 1})$.
        \item For every $x \in \bool^n$, $f(x^{i \to 0}) \geq f(x^{i \to 1})$.
    \end{enumerate}
    The function $f$ itself is said to be unate if it is unate in every coordinate $i \in [n]$. Unate functions are exactly those functions that become monotone under negation of some subset of the input variables. While these properties were defined over the domain $\bool^n$, they have clear analogues for functions over the domain $\pmone^n$. As mentioned earlier, we will switch between the two bases to aid the presentation of our results. 
    
    We assume basic familiarity with notions in complexity such as Boolean formulae, but review some terminology. An $\AC$ formula, refers to a formula with $\lnot$ gates and unbounded fan-in $\lor$ and $\land$ gates. A $\TC$ formula, refers to a formula where the gates are arbitrary linear threshold functions of unbounded fan-in. $\AC[\parity]$ and $\TC[\parity]$ are defined by allowing unbounded fan-in parity gates to $\AC$ and $\TC$, respectively. Finally, a formula is \emph{read-once} if every variable feeds into at most one gate. Recall that every gate in a formula has fan-out at most $1$.
    
    For a more comprehensive introduction to topics in the analysis of Boolean functions see \cite{saks_1993,odonnell2021analysis}.

\subsection{Polynomials}

    As described previously, each Boolean function can be represented uniquely as a formal multilinear polynomial through its Fourier representation.  We define the Fourier \emph{degree} (or simply degree) of $f$ as $\deg(f) = \max\{\abs{S}: \hat{f}(S) \neq 0\}$. We can extend this notion to polynomials that pointwise approximate $f$:
    
\begin{definition}
    Let $D \subseteq \pmone^n$ and $f \colon D \to \pmone$. A polynomial $p:\pmone^n \to \R$ is said to $\eps$-approximate $f$ if:
    \begin{enumerate}[label=(\arabic*)]
        \item $\forall x \in D$, $|p(x)-f(x)| \leq \eps$.
        \item $\forall x \in \pmone^n$, $|p(x)| \leq 1$.
    \end{enumerate}   
    The \emph{$\epsilon$-approximate degree} of $f$, denoted $\apxdeg_\epsilon(f)$, is defined as the minimum degree of any polynomial that $\epsilon$-approximates $f$. 
    If $\eps$ is not specified, we define the \emph{approximate degree} of $f$, denoted $\apxdeg(f)$, to be $\apxdeg_{1/3}(f)$.
\end{definition}

In this paper, we are primarily concerned with representations of $f$ via rational polynomials. This gives rise to a measure known as \emph{rational degree}, which is formally defined as follows.

\begin{definition}
Let $D \subseteq \pmone^n$ and $f \colon D \to \pmone$. If $p \colon D \to \mathbb{R}$ and $q \colon D \to \mathbb{R}$ are polynomials such that 
\[
\forall x \in D, \quad \abs{f(x) - \frac{p(x)}{q(x)}} \leq \epsilon,
\]
we say that $p/q$ is an $\epsilon$-approximate rational representation of $f$. The $\epsilon$-\emph{approximate rational degree} of $f$, denoted $\rdeg_\epsilon(f)$, is defined as the minimum value of $\max\{\deg(p), \deg(q)\}$ such that $p/q$ is an $\epsilon$-approximate rational representation of $f$. The \emph{rational degree} of $f$, denoted $\rdeg(f)$, is defined as $\rdeg_0(f)$.
\end{definition}

In the case of $\eps = 0$, we drop the \say{approximate} and refer to $p/q$ as a rational representation of $f$.
Unlike in the definition of approximate degree, there is no requirement for an approximate rational representation to be bounded outside of $D$. Whether or not such a boundedness condition is imposed matters significantly for the degree (see \cite{bunkotharithaler2020polystrikesback}) but not for the rational degree, as we outline in \cref{app:rdeg_postq}.

\subsection{Sensitivity and Certificate Complexity}
We now define some useful combinatorial measures of Boolean function complexity. Let $f$ be a Boolean function, $x\in \pmone^n$, and $B \subseteq [n]$. We say that $B$ is a sensitive block of $f$ at $x$ if $f(x) \neq f(x^B)$ where $x^B$ denotes the bitstring obtained by flipping all bits of $x$ indexed by $B$. We define, and denote by $
\bs_f(x)$, the \emph{block sensitivity of $f$ at $x$} as the maximum number of disjoint blocks that are all sensitive at $x$. By restricting our attention to sensitive blocks that are singletons we obtain the analogous notion of the \emph{sensitivity of $f$ at $x$}, denoted $\sens_f(x)$. The \emph{block sensitivity} of $f$ is defined as $\bs(f) = \max_{x \in \pmone^n} \bs_f(x)$. Similarly the \emph{sensitivity} of $f$ is defined as $\sens(f) = \max_{x \in \pmone^n} \sens_f(x)$. 
For $b\in \{0,1\}$, we also write $\sens^{(b)}(f) = \max_{x\in f^{-1}(b)} \sens_f(x)$.

A partial assignment is some function $\rho\colon [n] \to \{-1, 1, \star\}$. We define, and denote by $\abs{\rho}$, the size of the partial assignment $\rho$ as cardinality of the set $\{i\in [n]: \rho(i) \neq \star\}$. We say that a partial assignment $\rho$ is \emph{consistent} with some $x\in \pmone^n$ if $x_i = \rho(i)$ for all $i$ with $\rho(i) \neq \star$. Given a Boolean function $f$ we denote by $f\vert_{\rho}$ the restriction of $f$ to the set of inputs $x\in \pmone^n$ that are consistent with $\rho$. Given $b \in \pmone$, we say that a partial assignment $\rho$ is a \emph{$b$-certificate for $f$} if $f\vert_{\rho}(x) = b$ for all $x \in \dom(f\vert_{\rho})$. The \emph{$b$-certificate complexity} of $f$ is defined as 
\begin{align*}
    \cert_b(f) = \max_{x\in f^{-1}(b)} \min\{\abs{\rho}: \rho \text{ is a $b$-certificate for $f$ consistent with $x$}\}.
\end{align*}
The \emph{certificate complexity of $f$} is defined as $\cert(f) = \max_{b\in \pmone} \cert_b(f)$.

\subsection{Sign and Non-Deterministic Degree}

For a Boolean function $f$ we say that a polynomial $p\colon \pmone^n \to \mathbb{R}$ is a \emph{sign representation} if $\sgn(p(x)) = f(x)$ for all $x\in \pmone^n$ and $p(x) \neq 0$ on the entire hypercube. The \emph{sign degree} of $f$ is defined as the minimum degree of any polynomial that sign represents $f$. Alon \cite{alonslicingthecube} and Anthony \cite{ANTHONY199591} have shown that all but a negligible fraction of $n$-bit Boolean functions have sign degree at least $n/2$. Later, O'Donnell and Servedio proved \cite{ODONNELL2008298} that almost every Boolean function has sign degree at most $n/2 + \bigO(\sqrt{n \log{n}})$. 

A less common but somewhat similar notion is that of a \emph{non-deterministic polynomial} introduced by de Wolf \cite{dewolf2000}. In this context, it is customary to consider Boolean functions using the $\{0,1\}^n \to \{0,1\}$ representation. 

We say that $p\colon \{0,1\}^n \to \mathbb{R}$ is a non-deterministic polynomial for $f\colon \{0,1\}^n \to \{0,1\}$ if $p(x) 
= 0$ if and only if $f(x) = 0$. An easy calculation establishes the following relationship between the rational and non-deterministic degrees
\begin{align*}
    \rdeg(f) = \max\{\ndeg(f), \ndeg(\bar{f})\}.
\end{align*}
Note that this implies $\rdeg(f) = \rdeg(\bar{f})$.
As mentioned in the introduction de Wolf conjectured that $\dt(f) \leq \bigO(\ndeg(f)\ndeg(\bar{f}))$ for all total Boolean functions. By showing that $\ndeg(f) \leq \cert_1(f)$, de Wolf also established the inequality $\rdeg(f) \leq \cert(f)$ \cite{dewolf2000}.

\subsection{Quantum Query Complexity and Post-selection}
We assume basic familiarity with concepts in quantum information. While we review some of these, we direct the reader to, e.g. \cite{nielsen-chuang}, for background.

Consider a Boolean function $f$ over a domain $D$. We say an $\eps$-error quantum algorithm computes $f$ if for all $x \in D$ it outputs a bit $a(x)$, corresponding to the measurement outcome of the first qubit, such that $\Pr[a(x) = f(x)] \geq 1 - \eps$.
$\bqp$ is the class of problems that have efficient (polynomial-time) quantum algorithms with error $1/3$ and $\eqp$ is the analogous class of zero-error algorithms.
We can also define complexity classes corresponding to quantum algorithms augmented with the power of \emph{post-selection}. 

\begin{definition} \label{def:postbqp}
    $\postbqp$ is the set of languages $L \subset \{0,1\}^{\star}$ for which there exists a polynomial time quantum algorithm that for all inputs $x \in \bool^{\star}$ outputs two bits $a(x),b(x)$ (corresponding to the measurement results of the first and second qubits, respectively) such that:
    \begin{enumerate}[label=(\arabic*)]
        \item $\Pr[b(x) = 1] > 0$.
        \item If $x\in L$, then $\Pr[a(x) = 1 | b(x) = 1] \geq 2/3$.
        \item If $x\not \in L$, then $\Pr[a(x) = 1 | b(x) = 1] \leq 1/3$.
    \end{enumerate}
    $\posteqp$ is the corresponding class of \emph{zero-error} algorithms with post-selection: that is, the language defined by replacing the $2/3$ and $1/3$ above with $1$ and $0$, respectively.
\end{definition}

Each of these computational complexity classes has an associated query measure. Formally, we say a function has query access to a string $w$ if it has black-box access to a unitary $U$ such that $U \ket{i}\ket{b}= \ket{i}\ket{b\oplus w_i}$. When the input $w$ encodes the truth table of a Boolean function $f$, we will often write this as $U \ket{x}\ket{b} = \ket{x}\ket{b\oplus f(x)}$, where $x \in \bool^n$. The number of calls an algorithm makes to the unitary $U$ is its \emph{query complexity}. By $\quantumquery_{\eps}(f)$ and $\quantumquery_E(f)$ we denote the query complexities of $\eps$-error and zero-error quantum algorithms, respectively.  $\postquantumquery_{\eps}(f)$ and $\postquantumquery_{E}(f)$ are defined analogously for quantum algorithms with post-selection. For simplicity of notation, $\quantumquery(f)$ and $\postquantumquery(f)$ are understood to correspond to $\eps = 1/3$. 

A seminal result by Beals \textit{et al.} gives a lower bound quantum query complexity using polynomials \cite{beals1998}. Formally, we have $\quantumquery_\eps(f) \geq \apxdeg_\eps(f)/2$ for all (possibly partial) Boolean functions $f$. As a special case, $\quantumquery_E(f) \geq \deg(f)/2$. This result gave rise to the so-called \emph{polynomial method} for quantum query lower bounds. Similarly, it was shown by Mahadev and de Wolf that $\postquantumquery_\eps(f) = \Theta(\rdeg_\eps(f))$ and $\postquantumquery_E(f) = \Theta(\rdeg(f))$ for total functions $f$ \cite{mahadev2014}. It is not difficult to extend this result for partial functions (see \cref{app:rdeg_postq}). Nonetheless, it is surprising that the result does still hold for partial functions since the analogous result for quantum query complexity and approximate degree was recently shown to be false in \cite{ambainis2023}.

\section{Lower Bounds on Rational Degree}

In this section, we present rational degree lower bounds for certain classes of Boolean functions. Our results constitute progress towards showing that rational degree is polynomially related to Fourier degree for total functions.

First, we establish the asymptotically tight lower bound $\rdeg(f) \geq (\deg(f) + 1)/3$ for symmetric functions. 
Next, we prove that the rational degree equals the sensitivity for monotone functions, which implies that $\sqrt{\deg(f)} \leq \rdeg(f)$ for such functions. This lower bound is also tight. These results become key in proving rational degree lower bound for read-once Boolean formulae. Finally, we show that almost all Boolean functions have rational degree at least $n/2 - \bigO(\sqrt{n})$.

\subsection{Symmetric Functions}
\label{subsec:symmetric}
Our first lower bounds are for (possibly partial) functions which are constant on a large number of Hamming slices.  This lemma will later be useful in obtaining an unbounded separation of rational degree from quantum query complexity (and thus approximate degree) in the case of partial functions.

\begin{lemma}
    \label{lem:rdeg-symmetric}
    Let $f$ be a (possibly partial) nonconstant Boolean function over input domain $D \subseteq \bool^n$ and define 
    \begin{equation}
    \begin{aligned}\label{eq:s0_s1}
     &S_0 \coloneqq \{k\in \{0,1,\dots,n\} \colon |x| = k \implies f(x) = 0\},\\
     &S_1 \coloneqq \{k \in \{0,1,\dots,n\} \colon |x| = k \implies f(x) = 1\}.
    \end{aligned}
    \end{equation}
    Then $\rdegnull(f) \geq \frac{1}{2}\max(|S_0|, |S_1|)$.
\end{lemma}In the above definition of $S_0$ and $S_1$, we require that $D$ contains the entirety of the relevant Hamming slices.

We will use the Minsky-Papert symmetrization technique, which converts a multivariate polynomial over $\bool^n$ to a univariate polynomial over $\R$ \cite{minskypapert}. Formally, given $p: \bool^n \to \R$ we define $P(k) \coloneqq \mathop{\mathbb{E}}_{|x| = k}[p(x)]$.

\begin{proof}
    Since $\rdeg(f) = \rdeg(\bar{f})$, we can assume without loss of generality that $|S_0| \geq |S_1|$. It suffices to show that $\rdeg(f) \geq |S_0|$. Let $f = p/q$ be a rational representation of $f$. We will also assume that $p$ is nonnegative, which we can do with at most a factor $2$ overhead in the rational degree. Indeed, if $f = p^\prime/q^\prime$ where $p^\prime$ and $q^\prime$ can take on negative values, we can take $f = p/q$ where $p = p^\prime q^\prime$ and $q = {q^\prime}^2$, increasing the rational degree by a factor of $2$. Since $f$ takes values in $\{0,1\}$ and $q$ is clearly positive, $p$ must always be nonnegative.
    
    Applying the Minsky-Papert symmetrization technique to $p(x)$, we obtain a univariate polynomial $P(k)$ such that $\deg(p) \geq \deg(P)$ and $P(k) = 0$ for any $k \in S_0$. On the other hand, there exists at least one $k \in [n]$ such that $P(k) \neq 0$, since $f$ is nonconstant and nonnegative.
    Thus $\deg(p) \geq \deg(P) \geq |S_0|$. Recall that we also incurred a multiplicative blowup of $2$ in the rational degree by assuming that $p$ is nonnegative. This results in an additional factor of $1/2$ in the lower bound. Since this argument holds for every rational representation of $f$, the result follows.
\end{proof}

For total Boolean functions, we can use a slightly more careful version of the argument for \cref{lem:rdeg-symmetric} to show the next proposition.
\begin{prop}
 \label{cor:symmetric-lower-bound}
    If $f\colon\bool^n \to \bool$ is symmetric and non-constant then $\rdeg(f) \geq (\deg(f)+1)/3$.
\end{prop}
\begin{proof}

Since $\deg(f)$ is always at most $n$, it suffices to show that $\rdeg(f)\geq (n+1)/3$.

Define $S_0$ and $S_1$ as in \cref{eq:s0_s1}. Since $f$ is total and symmetric, we have 
\begin{equation}\label{eq:constaint}
    \abs{S_0}+\abs{S_1} = n+1.
\end{equation}

We analyze $f$ according to the following four cases.
\begin{enumerate}
    \item $f(0^n) = 0$ and $f(1^n) = 0$. In this case, $\ndeg(\bar{f}) \geq \abs{S_1}$. Importantly, there is no factor of $1/2$ since the symmetrized version of any non-deterministic polynomial for $\bar{f}$ is automatically non-zero at inputs $0$ and $n$. In addition, $\ndeg(f)\geq \abs{S_0}/2$ by considering the square of the symmetrized polynomial (like in the proof of \cref{lem:rdeg-symmetric}). Therefore $\rdeg(f) \geq \max(\abs{S_0}/2,\abs{S_1}) = \max(\abs{S_0}/2, n + 1 - \abs{S_0})$. This quantity is minimized when $|S_0| = \frac{2(n+1)}{3}$, from which we obtain a lower bound of $\rdeg(f) \geq (n+1)/3$.
    \item $f(0^n) = 0$ and $f(1^n) = 1$. In this case, regardless of whether we consider $f$ or $\bar{f}$, we will have at least one Hamming slice that is nonzero upon symmetrization. As such,
    $\rdeg(f)\geq \max(\abs{S_0},\abs{S_1})$ and so $\rdeg(f)\geq (n+1)/2$. 
    \item $f(0^n) = 1$ and $f(1^n) = 0$. We obtain a lower bound $\rdeg(f)\geq \max(\abs{S_0},\abs{S_1} \geq (n+1)/2$ in the exact same way as in case 2.
    \item $f(0^n) = 1$ and $f(1^n) = 1$. We follow the same reasoning as case 1, giving us the lower bound $\rdeg(f)\geq \max(\abs{S_0},\abs{S_1}/2) \geq (n+1)/3$. 
\end{enumerate}
In all cases, we have $\rdeg(f) \geq (n+1)/3$, which gives the proposition.
\end{proof}

In fact, \cref{cor:symmetric-lower-bound} can be tight even in its constant factor as witnessed by the following function. For positive integers $n$ divisible by $3$, let $\MT_n \colon \{0,1\}^n \to \{0,1\}$ be the Boolean function such that $\MT_n(x) = 1$ iff $\abs{x} \in \{n/3,n/3+1,\dots, 2n/3\}$. (The name $\MT$ abbreviates ``middle third''.) We will show in the next proposition that $\rdeg(\MT_n) \leq n/3+1$. But the polynomial degree of any non-constant symmetric Boolean function is at least $n - O(n^{0.548}) = n(1-o(1))$ as shown by von Zur Gathen and Roche \cite{vongathen_roche}. Therefore, $\rdeg(\MT_n) \leq \frac{1}{3}(1+o(1)) \cdot \deg(\MT_n)$.

\begin{prop}\label{prop:rdeg_lower_mt}
    $\rdeg(\MT_n) \leq n/3+1$.
\end{prop}

It is easy to see that $\ndeg(\overline{\MT_n}) \leq n/3+1$.
The harder part of this proposition is proving $\ndeg(\MT_n) \leq n/3$, which we do via the following sequence of three, increasingly stronger, lemmas. The first of these contains the key dimension-counting idea.
\begin{lemma}\label{lem:mt_rdeg_1}
    For every $k\in \{n/3,\dots, 2n/3\}$, there exists $z\in \{0,1\}^n$ with $\abs{z} = k$ and a multilinear polynomial $q$ of degree at most $n/3$ such that $q$ is non-zero at $z$ but is zero at all $x\in \{0,1\}^n$ with $\abs{x} \in \{0,1,\dots, n/3-1\}\cup \{2n/3+1,\dots, n\}$. 
\end{lemma}
\begin{proof}
    We prove this by a dimension-counting argument based on the following elementary fact whose easy proof we defer to the end of this subsection.
    \begin{fact}\label{fact:binom_fact}
    For every positive integer $n$ that is divisible by $3$,
    \begin{equation}
        \sum_{i\in \{0,1,\dots, n/3-1\}\cup \{2n/3+1,\dots, n\}} {\textstyle\binom{n}{i}} < \sum_{i\in \{0,1,\dots, n/3\}}{\textstyle\binom{n}{i}}.
    \end{equation}
    \end{fact}

    From this fact and dimension counting, we deduce there must exist a \emph{non-zero} multilinear polynomial $q$ of degree at most $n/3$ that vanishes at all $x\in \{0,1\}^n$ with $\abs{x} \in \{0,1,\dots, n/3-1\}\cup \{2n/3+1,\dots, n\}$. It remains to argue that $q$ is non-zero at some $z\in \{0,1\}^n$ with $\abs{z} = k$.

    Suppose for contradiction that $q$ is zero at all $z\in \{0,1\}^n$ with $\abs{z} = k$. Since $q$ is non-zero (as a polynomial) and multilinear, $q$ must be non-zero at some point $y\in \{0,1\}^n$ with $\abs{y} \in \{n/3,\dots, 2n/3\}$. Write $l$ for $\abs{y}$ so that $l\neq k$. We assume that $k>l$ as the argument in the case $k<l$ is analogous.

    We can obtain a contradiction by symmetrizing a ``translated version'' of $q$ as follows. Assume that $y$ is of the form $1^{l}0^{n-l}$. It is easy to see that the following argument makes this assumption without loss of generality. Consider the polynomial 
    \begin{equation}
        \tilde{q}(x_{l+1},\dots,x_n) \coloneqq q(1,\dots, 1, x_{l+1},\dots, x_n).
    \end{equation}
    Clearly, $\deg(\tilde{q})\leq \deg(q) \leq n/3$. But 
    \begin{align*}
        &\tilde{q}(0^l) \neq 0,
        \\
        &\tilde{q}(x) = 0 &&\text{for all $x\in \{0,1\}^n$ with $\abs{x} \in \{k-l\}\cup \{2n/3+1-l,\dots, n-l\}$},
    \end{align*}
    where the not-equal-to uses $\tilde{q}(0^l) = q(y)\neq 0$ and the equalities use the fact that $q$ vanishes at all $x\in \{0,1\}^n$ with $\abs{x} \in \{2n/3+1,\dots ,n\}$.

    Therefore, $\tilde{q}$ is a non-zero univariate polynomial vanishing at $n/3+1$ points, so the standard symmetrization argument implies $\deg(\tilde{q}) \geq n/3+1$, which contradicts $\deg(\tilde{q})\leq n/3$.
\end{proof}

The preceding lemma can be strengthened to:
\begin{lemma}\label{lem:mt_rdeg_2}
    For every $z\in \{0,1\}^n$ with $\abs{z} \in \{n/3,\dots, 2n/3\}$ there exists a multilinear polynomial $q$ of degree at most $n/3$ such that $q$ is non-zero at $z$ but is zero at all $x\in \{0,1\}^n$ with $\abs{x} \in \{0,1,\dots, n/3-1\}\cup \{2n/3+1,\dots, n\}$. 
\end{lemma}
\begin{proof}
    By the preceding lemma (\cref{lem:mt_rdeg_1}), there exists a polynomial $p$ of degree at most $n/3$ such that $p$ is non-zero at some $w\in \{0,1\}^n$ with $\abs{w} = \abs{z}$ but is zero at all $x\in \{0,1\}^n$ with $\abs{x} \in \{0,1,\dots, n/3-1\}\cup \{2n/3+1,\dots, n\}$. 

    Now, we simply observe that the following polynomial $q$ satisfies the conditions of the lemma:
    \begin{equation}
        q(x_1,\dots,x_n) = p(x_{i_1},\dots,x_{i_n}),
    \end{equation}
    where $i_1,\dots, i_n$ is any permutation of $1,\dots, n$ that bijects $\{i\in [n]\mid z_i =1\}$ with $\{i\in [n]\mid w_i=1\}$, which have the same size since $\abs{w} =\abs{z}$.
\end{proof}

The preceding lemma can be further strengthened to:
\begin{lemma}\label{lem:mt_rdeg_3}
    There exists a multilinear polynomial $q$ of degree at most $n/3$ that is non-zero at \emph{all} $z\in \{0,1\}^n$ with $\abs{z} \in \{n/3,\dots 2n/3\}$ but is zero at all $x\in \{0,1\}^n$ with $\abs{x} \in \{0,1,\dots, n/3-1\}\cup \{2n/3+1,\dots, n\}$. 
\end{lemma}

\begin{proof}
    The desired polynomial $q$ can be constructed by linearly combining the polynomials in $\{q_z \mid z\in \{0,1\}^n$ with $\abs{z} \in \{n/3,\dots, 2n/3\}\}$ with the property that $q_z$ is of degree at most $n/3$, is non-zero at $z$, and is zero at all $x\in \{0,1\}^n$ with $\abs{x} \in \{0,1,\dots, n/3-1\}\cup \{2n/3+1,\dots, n\}$. The existence of such $p_z$s is the content of  the preceding lemma (\cref{lem:mt_rdeg_2}). The existence of such a linear combination follows immediately from \cref{lem:avoidance_lemma}.
\end{proof}

\cref{prop:rdeg_lower_mt} is an easy corollary of the lemmas above.
\begin{proof}[Proof of \cref{prop:rdeg_lower_mt}]
    Since $\rdeg(\ndeg(\MT_n)) = \max(\ndeg(\MT_n),\ndeg(\overline{\MT_n}))$, it suffices to show $\ndeg(\MT_n) \leq n/3$ and $\ndeg(\overline{\MT_n}) \leq n/3+1$. The first inequality is just a rephrasing of \cref{lem:mt_rdeg_3}. 

    To see the second inequality, observe that the following polynomial $p$ is non-zero at every $x\in \{0,1\}^n$ with $\abs{x} \in \{0,1,\dots, n/3-1\}\cup \{2n/3+1,\dots, n\}$ and is zero at every $x\in \{0,1\}^n$ with $\abs{x} \in \{n/3,\dots,2n/3\}$:
    \begin{equation}
        p(x) \coloneqq \prod_{i\in \{n/3,\dots,2n/3\}} (x_1 + \dots + x_n - i).
    \end{equation}
    Therefore, $p$ non-deterministically represents $\overline{\MT_n}$. But the degree of $p$ is clearly $n/3+1$. Therefore, $\ndeg(\overline{\MT_n}) \leq n/3+1$.
\end{proof}

For completeness, we prove the fact used in the proof of \cref{lem:mt_rdeg_1}.
\begin{proof}[Proof of \cref{fact:binom_fact}]
    By symmetry of the binomial coefficients, the fact is equivalent to
     \begin{equation}
        \sum_{i\in \{0,1,\dots, n/3-1\}} {\textstyle\binom{n}{i}} < {\textstyle\binom{n}{n/3}}.
    \end{equation}
    Now, for every $i\in \{0,1,\dots,n/3-1\}$, we have 
    \begin{equation}
        \frac{\binom{n}{i+1}}{\binom{n}{i}} = \frac{n-i}{i+1},
    \end{equation}
    which decreases with $i$ and is at least $(n-(n/3-1))/(n/3) = 2+3/n>2$
    and so 
    \begin{equation}
        \sum_{i\in \{0,1,\dots, n/3-1\}} {\textstyle\binom{n}{i}} < (1/2 + 1/4 + \dots + 1/2^{n/3}) {\textstyle\binom{n}{n/3}} <  {\textstyle\binom{n}{n/3}},
    \end{equation}
    as claimed.
\end{proof}

\subsection{Monotone and Unate Functions}
\label{subsec:monotone}
In this subsection we prove that $\rdeg(f) = \sens(f)$ for monotone and unate Boolean functions $f$.  We note that for monotone functions it suffices to prove that $\sens(f) \leq \rdeg(f)$. This is because the certificate complexity of a monotone Boolean functions $f$ equals its sensitivity $\cert(f) = \sens(f)$ \cite{nisan1991}. Combining this with the fact that $\rdeg(f) \leq C(f)$ (via \cite{dewolf2000}) we can already conclude the other inequality.

\begin{claim}\label{claim:rdegmonotone}
For monotone Boolean functions $f\colon\{0,1\}^n\to \{0,1\}$, 
$s(f) \leq \rdeg(f)$.
\end{claim}
Our proof is similar to the proof that for all monotone Boolean functions $s(f) \leq \deg(f)$ as presented in \cite[Proposition 4]{buhrmandewolf2002survey}.

\begin{proof}
Suppose without loss of generality that $f$ is monotone increasing. We prove the claim by showing that
\begin{equation}
    \senszero(f) \leq \ndeg(\bar{f}) \quad \text{and} \quad \sensone(f) \leq \ndeg(f).
\end{equation}

We only prove the first inequality as the second can be proven analogously. Let $x$ be such that $\senszero(f) = \sens_f(x)$. All sensitive variables must be $0$ in $x$ since $f$ is monotone increasing. Moreover, setting any sensitive variable to $1$ changes the value of $f$ from $0$ to $1$. Therefore, fixing all variables in $x$ except for the $\senszero(f)$ many sensitive variables yields the $\OR_m$ function on $m\coloneqq\senszero(f)$ variables. Since $\ndeg(\overline{\OR}_m) \geq m$, $\ndeg(\bar{f}) \geq \senszero(f)$.
\end{proof}

We remark that \Cref{claim:rdegmonotone} cannot be extended to all Boolean functions, as evidenced by the Kushilevitz function $K_m\colon \{0,1\}^{6^m} \to \{0,1\}$ \cite{nisanwigderson1995}, which has full sensitivity, but degree $3^m$. Since $\sens(f) = \cert(f)$ and $\sqrt{\deg(f)} \leq \sens(f)$ for monotone functions \cite{nisan1991}, we have the following corollary.

\begin{corollary}\label{cor:rdeg-monotone}
    For monotone Boolean functions $f$, $\rdeg(f) = \sens(f)$. In particular, $\sqrt{\deg(f)} \leq \rdeg(f)$.
\end{corollary}

This result also extends easily to unate functions. We first show that negating inputs does not change the rational degree.
\begin{claim}\label{prop:negation}
    Let $f$ be a Boolean function and let $S \subseteq [n]$. Let $f^\prime$ be the Boolean function defined by applying $f$ to the input after negating the inputs in $S$. Then $\rdeg(f) = \rdeg(f^\prime)$.
\end{claim}
\begin{proof}
    Let $p,q:\bool^n \to \R$ be polynomials such that $f(x) = p(x)/q(x)$ for all $x \in \bool^n$. Define a polynomial $p^\prime$ by taking every $i \in S$ and replacing every instance of $x_i$ in the multilinear expansion of $p$ with $1-x_i$, and define $q^\prime$ similarly. This operation does not change the degree: $\deg(p) = \deg(p^\prime)$ and $\deg(q) = \deg(q^\prime)$. Furthermore, it can easily be seen that $f^\prime(x) = p^\prime(x)/q^\prime(x)$, and the result follows.
\end{proof}

This fact allows for a simple proof that unate functions have large rational degree.

\begin{corollary}\label{cor:rdeg-unate}
    For unate Boolean functions $f$, $\sqrt{\deg(f)} \leq \rdeg(f)$.    
\end{corollary}
\begin{proof}
    Every unate Boolean function is equivalent to a monotone function under a negation of some subset of the input variables. Thus the claim follows by combining \cref{cor:rdeg-monotone} and \cref{prop:negation}. 
\end{proof}

Note that these bounds are tight, as witnessed by the $\AND$-of-$\OR$s function, $f = \AND_n \circ \OR_n$, on $n^2$ bits, which has $\rdeg\lpar f\rpar \leq \cert(f) = n = \sqrt{\deg(f)}$. Writing $f(x) = \bigwedge_{i=1}^n \bigvee_{j=1}^n x_{ij} $, then an explicit rational representation of $f$ of degree $n$ is
\begin{equation}
\frac{\prod_{i=1}^n(\sum_{j=1}^n x_{ij})}{\prod_{i=1}^n(\sum_{j=1}^n x_{ij}) + \sum_{i=1}^n\prod_{j=1}^n(1-x_{ij})}.
\end{equation}

\subsection{Read-Once Formulae}\label{subsec:readonce}

In this section, we demonstrate lower bounds on read-once $\AC, \AC[\oplus], \TC,$ and $\TC[\oplus]$ formulae. Our result is actually slightly stronger, and holds for any read-once formula where the gates are functions of high rational degree. We begin by proving a composition result for rational degree.
\begin{lemma}\label{lem:trickle-down}
    Let $f \colon \bool^n \to \bool$ and $g_i \colon \bool^{n_i} \to \bool$ be Boolean functions where every variable in each function is relevant. Defining $h\colon \{0,1\}^{\sum n_i}\to \{0,1\}$ to be $h(x^1,\dots, x^n) =  f(g_1(x^1),\dots,g_n(x^n))$, we have that 
    \begin{align*}
        \rdeg(h)\ge \max\{\rdeg(f), \rdeg(g_1), \dots, \rdeg(g_n)\}.
    \end{align*}
\end{lemma}

\begin{proof}
    Since every variable is relevant for each $g_i$, we know there exist restrictions $\rho^i$ to all but $1$ variable in each $x^i$ such that $g_i|_{\rho^i}(x^i) = x_{k_i}^i$ or $(1-x_{k_i}^i)$ for some $1\le k_i\le n_i$. Considering the restriction $\rho = \rho^1\cup\dots\cup \rho^n$, it is evident that $h|_\rho(x) = f(x^1_{k_1},\dots,x^n_{k_n})$ up to negations. Therefore, 
    \begin{align}\label{eqn:f}
        \rdeg(h)\ge \rdeg(h|_\rho)\ge \rdeg(f).
    \end{align} \ 

    Now pick an arbitrary $i$. Since, by assumption, every variable of $f$ is relevant, there exists an assignment $x_j=z_j$ for all $j\neq i$ such that $f(z_1,\dots z_{i-1},x_i,z_{i+1},\dots, z_n) = x_i$ or $\conj{x_i}$. Since $g_i$ is nonconstant, it follows that there exists an assignment to the variables $(x^j)_{j\neq i}$ such that each $g_j(x^j)$ is fixed to $z_j$. 
    
    Let $\tau$ be the restriction induced by this partial assignment. Then  \begin{align*}
        h|_{\tau}(x) = f(z_1,\dots z_{i-1}, g_i(x^i), z_{i+1},\dots z_n) = g_i(x^i)\text{ or }\conj{g_i(x^i)}.
    \end{align*}
    Consequently, 
    \begin{align}\label{eqn:gi}
        \rdeg(h)\ge \rdeg(h|_{\tau})\ge \rdeg(g_i).
    \end{align} 
    Combining \Cref{eqn:f,eqn:gi} gives the desired result. 
\end{proof}

Next, we show how this composition result can be applied to rational degree lower bounds for read-once formulae.

\begin{lemma}
\label{lem:branching-factor}
    Let $f$ be written as a read-once formula with symmetric gates where the maximum branching factor of any node is $w$. Then $\rdeg(f) = \Omega(w)$. 
\end{lemma}
\begin{proof}
    Let $p/q$ be a rational representation of $f$. Due to \cref{prop:negation}, we can assume without loss of generality that only the literals $x_i$, appear in the formula and not $\overline{x_i}$.
    
    Now consider the node $G$ with branching factor $w$. Let $F$ be the subformula with top gate $G$ and let $F_1,\dots, F_w$ be the read-once subformulae below $G$. Each $F_i$ is nonconstant, which implies the existence of a restriction $\rho_i$ of all but $1$ variable in each $V_i$ such that toggling the sole live variable (say $x_{k_i}$) toggles the value of $F_i$ (i.e. $F_i|_{\rho_i} = x_k$ or $\conj{x_k}$ for some $k$). As the $V_i$ are disjoint (as $f$ is read-once), these restrictions together define a unified restriction $\rho$ such that $F|_\rho(x) = G(x_{k_1},\dots ,x_{k_w})$ up to negations. Inductively using \Cref{lem:trickle-down} on the formula $f|_\rho$ by starting at the top node and going down the path to $G$, it follows that 
    \begin{equation}
        \rdeg(f)\geq \rdeg(f|_\rho)\ge \rdeg(F|_\rho) =\rdeg(G) \geq w/2,
    \end{equation} 
    where the last inequality follows from \Cref{lem:rdeg-symmetric}.
\end{proof}
Now we can prove rational degree lower bounds on read-once formulae with symmetric gates.
\begin{corollary}\label{cor:read-once-symmetric}
    Let $f$ be written as a depth-$d$ read-once formula with symmetric gates. Then $\rdegnull(f) = \Omega(\deg(f)^{1/d})$.
\end{corollary}
\begin{proof}
    The result follows from \Cref{lem:branching-factor} and a simple contradiction argument: if all nodes have branching factor strictly less than $n^{1/d}$ then there must be strictly fewer than $n$ literals. Note that $n \geq \deg(f)$ so the lower bound $\Omega(n^{1/d})$ is stronger.
\end{proof}
In particular, this result holds for read-once $\AC[\oplus]$ formulae. This lower bound is tight for $d=2$, as witnessed by the $\AND$-of-$\OR$s function $f = \AND_{\sqrt{n}}\circ \OR_{\sqrt{n}}$. Indeed, $\rdeg\lpar f\rpar \leq \cert(f) = n^{1/2} = \sqrt{\deg(f)}$. 

A similar, albeit weaker result can be shown for read-once formulae where the gates are symmetric \emph{or} unate.
\begin{corollary}\label{cor:read-once-unate}
    Let $f$ be written as a depth-$d$ read once formula where every gate is symmetric and/or unate. Then $\rdeg(f) = \Omega(\deg(f)^{1/2d})$.
\end{corollary}
\begin{proof}
    The proof is similar to that of \cref{cor:read-once-symmetric}. Consider any gate in the formula and let $w$ be its branching factor. If this node corresponds to a symmetric or unate gate then we have $\rdeg(f) = \Omega(w)$ (via \cref{lem:rdeg-symmetric}) or $\rdeg(f) = \Omega(\sqrt{w})$ (via \cref{cor:rdeg-unate}), respectively. 
    By contradiction, at least one node has branching factor at least $\deg(f)^{1/d}$.
    The result follows by applying \cref{lem:branching-factor}.
\end{proof}
Thus, in particular, read-once $\TC[\oplus]$ formulae have rational degree at least $\Omega(\deg(f)^{1/2d})$, as they consist of linear threshold gates, which are unate, and parity gates, which are symmetric.

Observe that if we restrict our attention to read-once $\AC$ and $\TC$ (that is, if we exclude parity gates), the lower bounds in \cref{cor:read-once-symmetric} and \cref{cor:read-once-unate} improve to $\sqrt{n}$ via \cref{cor:rdeg-unate}. This is because the composition of unate functions is unate, and read-once $\AC$ and $\TC$ are entirely comprised of unate gates.

\begin{corollary}\label{cor:threshold_read_nce}
    For any read-once $\TC$ formula $f$ on $m$ disjoint variables, $\sqrt{m} \leq \rdeg(f)$.    
\end{corollary}
\begin{proof}
    We know that $f$ is composed entirely of unate gates. By a simple inductive argument, the composition of unate functions is unate. Thus we deduce that $f$ is unate, and since the degree of $f$ is $m$, the result follows from \cref{cor:rdeg-unate}.
\end{proof}
In \cref{app:readonce_tc}, we prove a structural result, \cref{prop:read_once_threshold}, showing that any read-once $\TC$ formula on $m$ disjoint variables must restrict to either an $\AND$ or $\OR$ function on at least $\sqrt{m}$ disjoint variables (up to negation). \cref{prop:read_once_threshold} immediately implies \cref{cor:threshold_read_nce} but is not implied by it and so may be of independent interest.

It can be shown that there exist arbitrary-depth read-once formulae with rational degree at most $\sqrt{n}$ (see \Cref{fig:and-or-upperbound}). The inverse dependence on depth in \cref{cor:read-once-symmetric} and \cref{cor:read-once-unate} is somewhat unintuitive, and we conjecture that these results are not tight for $d > 2$.

    \begin{figure}[H]
        \centering        \includegraphics[width=4.5in]{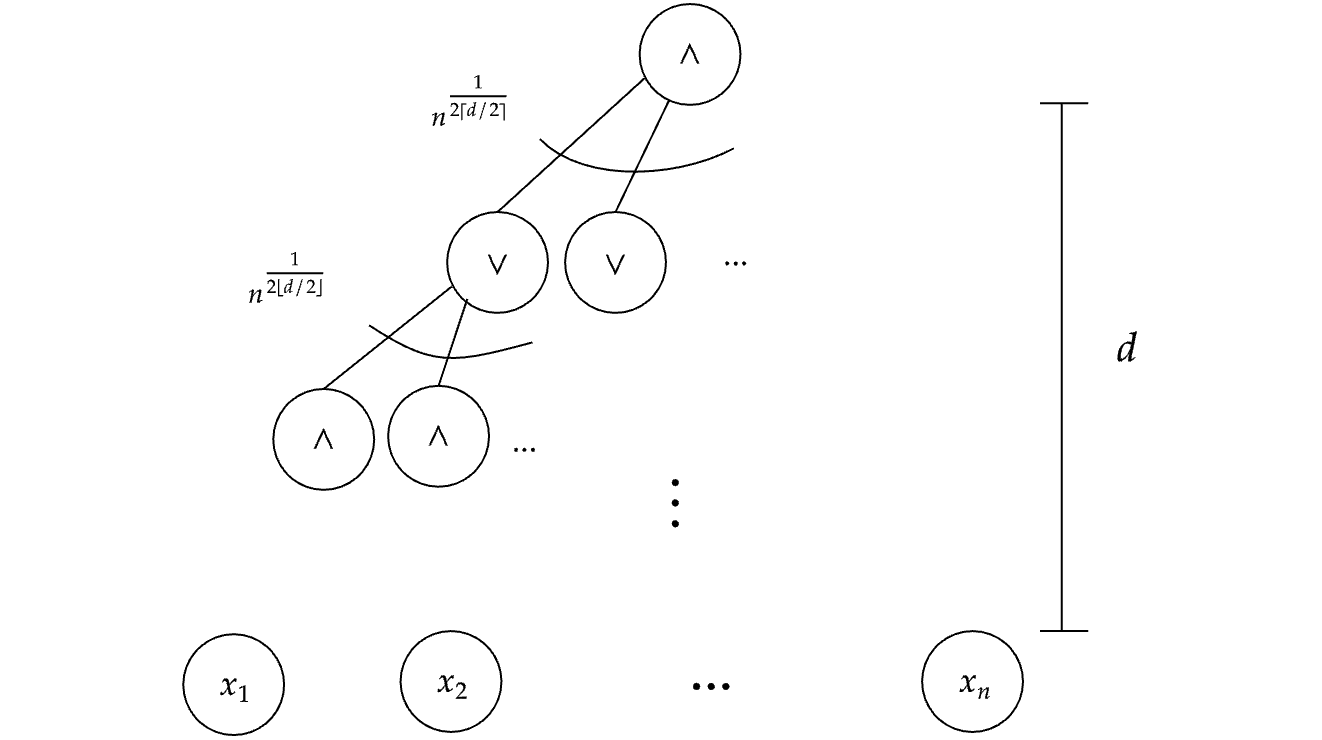}
        \caption{A depth-$d$ $\AND$-$\OR$ tree with certificate complexity $\sqrt{n}$, and thus rational degree at most $\sqrt{n}$. Indeed, setting all input wires to $1$ for $\AND$ functions and a single input wire to $1$ for $\OR$ functions along a single path to root gives a $1$ certificate, and setting all input wires to $0$ for $\OR$ functions and a single input wire to $0$ for $\AND$ functions gives a $0$-certificate. One can easily verify that both of these certificates are of size $\sqrt{n}$.}
        \label{fig:and-or-upperbound}
    \end{figure}

\subsection{Random Functions}
\label{subsec:random}
As our final piece of evidence that rational degree is polynomially related to degree, we prove that all but a negligible fraction of Boolean functions $f\colon \pmone^n \to \pmone$ have rational degree at least $n/2 -\bigO(\sqrt{n})$. 

As mentioned in the introduction Alon \cite{alonslicingthecube} and Anthony \cite{ANTHONY199591} used counting arguments to show that all but a negligible fraction of $n$-bit Boolean functions have sign degree at least $n/2$. We note that while sign degree is asymptotically a lower bound for rational degree, it loses a factor of 2, which gives a lower bound of $n/4$. To get a sign representation of $f$ from the nondeterministic polynomials $p(x)$ for $f$ and $q(x)$ for $\bar{f}$, one can construct $p(x)^2 - q(x)^2$. To show the tighter result for rational degree, we restate a variant of the function counting theorem used by Anthony.

Given a finite set $X$ and a mapping $\phi\colon X \to \mathbb{R}^d$, we say that a $\phi$-separable dichotomy of $X$ is a partition of $X$ into subsets $X^{+}\cup X^{-}$ such that there exists some $w\in \mathbb{R}^d$ for which $w\cdot \phi(x) > 0$ for all $x\in X^{+}$ and $w\cdot \phi(x) < 0$ for all $ x\in X^{-}$.

\begin{theorem}[Function counting theorem, \cite{cover1965}]\label{thm:function_counting}
    Let $\phi\colon S \to \mathbb{R}^d$. Let $X = \{x_1, \dots, x_N\} \subseteq S$. If a $\phi$-surface (i.e., a set of the form $\{x \in S: w\cdot \phi(x) = 0\}$ for some $w\in \mathbb{R}^d$) contains a set of points $Y = \{y_1, y_2, \dots, y_k\}\subseteq S$, where $\phi(y_i)$ are linearly independent for all $i$, and where the projection of $\phi(x_1), \dots, \phi(x_N)$ onto the subspace orthogonal to the span of the $\phi(y_i)$'s is in general position, then there are $C(N, d-k)$ many $\phi$-separable dichotomies of $X$, where 
    \begin{align*}
        C(N, d) = 2 \sum_{i=0}^{d-1} \binom{N-1}{i}.
    \end{align*}
\end{theorem}

We consider the following adaptation of the above theorem. Consider a set of $N$ points $S = \{v_1, \dots, v_n\}$ in $\mathbb{R}^D$. Given a 2-coloring of the points $f\colon [N] \to \{-1,1\}$, we say that the coloring $f$ is separable by two hyperplanes if there exist hyperplanes $H_j = \{v:\ \alpha_j\cdot v = 0 \}$ for $j = 1,2$ such that 
\begin{align*}
    \forall i \in [N]\colon\ f(i) = \sgn((\alpha_1 \cdot v_i)(\alpha_2 \cdot v_i)).
\end{align*}

\begin{corollary}\label{cor:two_plane_colorings}
    Given $N$ points in $\mathbb{R}^M$, the number of two colorings $f\colon [N] \to \pmone$ that are separable by two hyperplanes is at most $C(N, M)^2$.
\end{corollary}
\begin{proof}
    Let $S \subset \mathbb{R}^M$ be given and suppose that $f$ is a coloring that is separated by the hyperplanes $H_1$ and $H_2$. Then there exist colorings $f_1, f_2$ that are separated by the hyperplanes $H_1$ and $H_2$ respectively. Since there are at most $C(N,M)$ choices for each of $f_1$ and $f_2$, the number of such colorings $f$ is bounded by $C(N,M)^2$.
\end{proof}

\begin{lemma}\label{lemma:rational_degree_count}
    Let $m \leq n$ be two positive integers. The number of Boolean functions $f\colon \pmone^n \to \pmone$ with rational degree at most $m$ is at most $C(2^n, \binom{n}{\leq m})^2$.
\end{lemma}
\begin{proof}
    Let $M = \binom{n}{\leq m}$. For each $x \in \pmone^n$ define $v_x \in \mathbb{C}^M$ by letting $(v_x)_S = \chi_S(x)$ for $S \subseteq [n]$, $\abs{S} \leq m$. Suppose $p/q$ is a rational representation of $f$ of degree at most $m$. Then 
    \begin{align*}
        f(x) = \sgn(p(x)/q(x)) = \sgn(p(x)q(x)) = \sgn((\hat{p}\cdot v_x)(\hat{q}\cdot v_x)).
    \end{align*}
    Thus the coloring given by $f$ is separated by the hyperplanes defined by $\hat{p}$ and $\hat{q}$. The result follows by \Cref{cor:two_plane_colorings}. 
\end{proof}

We can now state and prove our extremal lower bound on rational degree. 
\begin{corollary}
\label{prop:random}
    All but a negligible fraction of Boolean functions on $n$ variables have rational degree at least $n/2 - \bigO(\sqrt{n})$.
\end{corollary}
\begin{proof}
    Write $m = n/2 - \sqrt{cn}$ for some constant $c$, and let $N = 2^n$. Then by \Cref{cor:two_plane_colorings} there are at most $C\bigl(N, \binom{n}{\leq m} \bigr)^2$ Boolean functions on $n$ variables of rational degree less than $m$. By the Chernoff bound $\binom{n}{\leq n/2-\lambda} < 2^n e^{-2\lambda^2/n}$ and the Hamming bound, we have that the proportion of Boolean functions with rational degree strictly less $m$ is bounded above by
    \begin{align*}
        \frac{C\bigl( N, \binom{n}{\leq m} \bigr)^2}{2^N} \leq \frac{C(N, Ne^{-2c})^2}{2^N} \leq O\bigl(2^{N(2h_2(e^{-2c})-1)}\bigr),
    \end{align*}
    where $h_2(\cdot)$ denotes the binary entropy function. Solving the inequality $h_2(e^{-2c}) < 1/2$ numerically we find that for $c \geq 1.104$, the above bound tends to $0$.
\end{proof}

\section{Composition Lemmas for Rational Degree}\label{sec:composition-lemmas}
Typically, the first composition lemma one proves for a Boolean complexity measure is exact composition with respect to $\XOR$. However, this fails for the non-deterministic degree, e.g., $\ndeg(x_1 \oplus x_2) = 1 \neq 2 = \ndeg(x_1) + \ndeg(x_2)$. Instead, we find that the non-deterministic degree composes exactly with respect to $\AND$.

In this section, we consider Boolean functions $f \colon \{0,1\}^m \to \{0,1\}$, and $g \colon \{0,1\}^n \to \{0,1\}$ for some $m,n \in \mathbb{N}$. By $f\wedge g$ and $f\vee g$ we denote the functions $\{0,1\}^m \times \{0,1\}^n \to \{0,1\}$ given by $(f\wedge g)(x,y) = f(x) \wedge g(y)$ and $(f \vee g)(x,y) = f(x) \vee g(y)$, respectively. We also write $\mathbb{R}[X]$, $\mathbb{R}[Y]$, and $\mathbb{R}[X,Y]$ for the rings of formal polynomials $\mathbb{R}[X_1,\dots,X_m]$, $\mathbb{R}[Y_1,\dots,Y_n]$, and $\mathbb{R}[X_1,\dots X_m,Y_1,\dots,Y_n]$, respectively.

\begin{prop}\label{prop:ndeg_and}
    Suppose $f$ and $g$ are non-constant Boolean functions on disjoint variables. Then 
    \begin{equation}
        \ndeg(f\wedge g) = \ndeg(f) + \ndeg(g).
    \end{equation}
\end{prop}

\begin{proof}
    Let $k \coloneqq \ndeg(f)$ and $l \coloneqq \ndeg(g)$. It is clear that $\ndeg(f\wedge g)\leq k+l$ by considering the product of non-deterministic polynomials for $f$ and $g$. We proceed to prove the lower bound.

    Let $P(X,Y) \in \mathbb{R}[X,Y]$ be a non-deterministic representation of $f\wedge g$. Since we will lower bound the degree of $P$, we may assume $P$ is multilinear without loss of generality. 

    Then, write
    \begin{equation}
        P(X,Y) = \sum_{S\subseteq  [m]} p_{S}(Y) \prod_{i \in S} X_i
        = \sum_{S\subseteq  [m]\colon \abs{S} \geq k} p_{S}(Y) \prod_{i \in S} X_i + \sum_{S\subseteq  [m] \colon \abs{S} < k} p_{S}(Y) \prod_{i \in S} X_i.
    \end{equation}

    Let $y\in g^{-1}(0)$. Then $P(X,y) \in \mathbb{R}[X]$ satisfies $P(x,y) = 0$ for all $x\in \{0,1\}^m$. This means that $P(X,y)$ is the $0$ polynomial by the uniqueness of exact multilinear polynomial representations of Boolean functions. Therefore $p_S(y) = 0$ for all $S\subseteq[m]$.

    Let $y \in g^{-1}(1)$. Then $P(X,y)$ is a non-deterministic representation of $f$. Therefore, since $\ndeg(f) = k > 0$, at least one entry of the vector $(p_S(y))_{S\subseteq [m], \abs{S}\geq k}$ must be non-zero. 

   We now appeal to the following lemma with $\calD$ set to $g^{-1}(1)$ and the $f_i$ functions given by (functions induced by) the set $\{p_S:\ S\subseteq [m], \abs{S}\geq k\}$. We defer the short proof of this lemma to after this proof.
   
\begin{lemma}\label{lem:avoidance_lemma}
    Let $n \in \mathbb{N}$. Let $\calD$ be a finite non-empty set. Let $f_1,\dots, f_n \colon \calD \to \mathbb{R}$. Suppose that for all $y \in \calD$, there exists $i\in [n]$ such that $f_i(y) \neq 0$. Then there exist $\alpha_1,\dots, \alpha_n >0 $ such that
    \begin{equation}
        (\alpha_1 f_1 + \dots + \alpha_n f_n)(y) \coloneqq \alpha_1 f_1(y) + \dots + \alpha_n f_n(y) \neq 0
    \end{equation}
    for all $y\in \calD$.
\end{lemma}

From the lemma, we see that for every $S\subseteq [m]$ with $\abs{S}\geq k$, there exists $\alpha_S > 0$ such that the polynomial:
    \begin{equation}
        Q(Y) \coloneqq \sum_{S\subseteq[m] \colon \abs{S} \geq  k } \alpha_S \cdot p_S(Y)  \in \mathbb{R}[Y]
    \end{equation}
    satisfies $y\in g^{-1}(1) \implies Q(y) \neq 0$.

    On the other hand, we have already observed that $y\in g^{-1}(0)$ implies $p_S(y) = 0$ for all $S\subseteq [m]$. Therefore, $y\in g^{-1}(0) \implies Q(y) = 0$. 
    
    We conclude that $Q(Y)$ non-deterministically represents $g$. Therefore, the degree of $Q$ is at least $l$. Thus, there must exist $S^*\subseteq [m]$ with $\abs{S^*} \geq k$ such that $\deg(p_{S^*}(Y)) \geq l$. 

    But $l > 0$ since $g$ is non-constant so the degree of the term
    \begin{equation}
        p_{S^*}(Y) \prod_{i \in S^*} X_i
    \end{equation}
    is at least\footnote{Note that if $l=0$, $p_S^*(Y)$ could be the zero polynomial in which case the degree of the term here is also zero.} $k+l$. Therefore, $P$ has degree at least $k+l$, as required.
\end{proof}

\begin{proof}[Proof of \cref{lem:avoidance_lemma}]
    Since $\calD$ is finite, for all $i\in [n]$, there exists $0< b_i < B_i$ such that $f_i(\calD) - \{0\} \subseteq (-B_i, -b_i) \cup (b_i, B_i)$. Now define $\alpha_1 = 1$ and for $i = 2,\dots, n$, define $\alpha_i > 0$ by 
    \begin{equation}
        \alpha_i b_i \coloneqq \sum_{j=1}^{i-1}\alpha_j B_j + 1 > \sum_{j=1}^{i-1}\alpha_j B_j.
    \end{equation}
    It is straightforward to verify that
    $(\alpha_1 f_1 + \dots + \alpha_n f_n)(y) \neq 0$ for all $y\in \calD$.
\end{proof}

We can also show the following composition lemma for the non-deterministic degree with respect to $\OR$. (Note that this lemma also does not hold with respect to $\XOR$ since $\ndeg((x_1\oplus x_2)\oplus x_3) = 2 \neq 1 = \max(\ndeg(x_1\oplus x_2), \ndeg(x_3))$ --- indeed we were unable to find any \say{clean} composition lemma for $\ndeg$ or $\rdeg$ with respect to $\XOR$.)
\begin{prop}\label{prop:ndeg_or}
    Suppose $f$ and $g$ are non-constant Boolean functions on disjoint inputs. Then
    \begin{equation}
    \ndeg(f \vee g) = \max(\ndeg(f), \ndeg(g)).
\end{equation}
\end{prop}

\begin{proof}
     Let $k \coloneqq \ndeg(f)$ and $l \coloneqq \ndeg(g)$. It is clear that $\ndeg(f \vee g) \geq \max(\ndeg(f), \ndeg(g))$ by restriction. We proceed to prove the upper bound.

     Suppose $p(X)\in \mathbb{R}[X]$ and $q(Y) \in \mathbb{R}[Y]$ are non-deterministic polynomials for $f$ and $g$ respectively such that $\deg(p(X)) =k$ and $\deg(q(Y)) = l$. Let $\alpha,\beta\in \mathbb{R} - \{0\}$ be such that
    \begin{align}
        &\forall x\in \{0,1\}^m, \abs{\alpha p(x)} \leq 1;
        \\
        &\forall x\in \{0,1\}^n \text{ s.t. } q(x) \neq 0, \abs{\beta q(x)} \geq  100.
    \end{align}
    
    Then it easy to verify that $\alpha p(X) + \beta q(Y) \in \mathbb{R}[X,Y]$ is a non-deterministic polynomial for $f \vee g$.
    Since $\deg(\alpha p(X) + \beta q(Y))\leq \max(k,l)$, the proposition follows.
\end{proof}

We have the following corollary for the rational degree of the $\AND$ and $\OR$ of Boolean functions.

\begin{corollary}\label{cor:rdeg_and_composition}
    Suppose $f$ and $g$ are non-constant Boolean functions on disjoint variables. Then
    \begin{align}
        \rdeg(f \wedge g) =& \max(\ndeg(f) + \ndeg(g), \ndeg(\overline{f}), \ndeg(\overline{g})),
        \\
        \rdeg(f \vee g) =& \max(\ndeg(\overline{f}) + \ndeg(\overline{g}), \ndeg(f), \ndeg(g)).
    \end{align}
\end{corollary}

\begin{proof}
This follows from \cref{prop:ndeg_and,prop:ndeg_or}, the identity $\rdeg(h) = \max(\ndeg(h),\ndeg(\overline{h}))$, and De Morgan's laws.
\end{proof}

\section{Rational Degree versus Spectral Sensitivity}\label{sec:rdeg-v-lambda}

The spectral sensitivity $\lambda(f)$ is a complexity measure of a Boolean function $f$ first defined formally by Aaronson \textit{et al.} \cite{aaronson2020degree}, although it has implicitly appeared earlier in Huang's proof of the Sensitivity Theorem \cite{huang2019induced}. Aaronson and his coauthors suggested \cite{aaronson2020degree} that it may be possible to establish a lower bound on the rational degree $\rdeg(f)$ by relating it to the spectral sensitivity $\lambda(f)$. In this direction, we show that $\lambda(f) \leq \rdeg(f)$ is false by exhibiting a degree $\sfrac{3}{2}$ separation between these two measures. The spectral sensitivity is defined formally as follows.

\begin{definition}
    Let $f\colon \{0,1\}^n \to \{0,1\}$ be a Boolean function. The sensitivity graph of $f$ is defined as the subgraph $G_f = (V, E)$ of the Boolean hypercube with $V = \{0,1\}^n$ and edge set ${E = \{(x,y): \abs{x\oplus y} = 1, f(x) \neq f(y)\}}$. The spectral sensitivity of $f$ is defined as the spectral norm of the adjacency matrix of $G_f$, ie. $\lambda(f) = \norm{A_f}$, where $A_f$ denotes the adjacency matrix of $G_f$.
\end{definition}

For an even positive integer $n$, let $\EH_n \colon \{0,1\}^n \to \{0,1\}$ be the Boolean function such that $\EH_n(x) = 1$ iff $\abs{x} = n/2$. As we show below, composing $\AND_n$ with the negation of $\EH_n$ gives a degree-$\sfrac{3}{2}$ separation between the spectral sensitivity and the rational degree. 

\begin{prop}
\label{prop:separation}
    The function $f = \AND_n \circ\> \overline{\EH_n}$ satisfies $\rdeg(f) = \Theta(n)$, $\lambda(f) = \Theta(n^{\sfrac{3}{2}})$, $\sens(f) = \Theta(n^2)$ and $\deg(f) = \Theta(n^2)$.
\end{prop}
\begin{proof}
    Since the spectral sensitivity behaves multiplicatively under composition (c.f. \cite[Theorem~29]{aaronson2020degree}) we have $\lambda(f) = \lambda(\AND_n)\cdot \lambda(\overline{\EH_n})$. It is not hard to see that $\lambda(\AND_n) = \sqrt{n}$ and $n/2 \leq \lambda(\overline{\EH_n}) \leq n$. Thus $\lambda(f) = \Theta(n^{\sfrac{3}{2}})$.

    On the other hand, by our composition theorem (\cref{prop:ndeg_and}) $\ndeg(f) = n \cdot \ndeg(\overline{\EH_n})$. By De Morgan's laws and \cref{prop:ndeg_or} $\ndeg(\overline{f}) = \ndeg(\OR_n \circ \EH_n) = \ndeg(\EH_n)$. So it remains to identify the non-deterministic degrees of $\overline{\EH_n}$ and $\EH_n$. To this end, observe that $q(x) = \sum_{i=1}^n x_i - n/2$ is a non-deterministic representation of $\overline{\EH_n}$ with degree $1$, which must also be degree-optimal since $\overline{\EH_n}$ is non-constant. Therefore, $\ndeg(\overline{\EH_n}) = 1$ and so $\ndeg(f) = n$. Moreover, we trivially have $\ndeg(\EH_n) \leq n$ so $\ndeg(\overline{f}) \leq n$. Therefore, $\rdeg(f) = \max(\ndeg(f), \ndeg(\overline{f})) = n$.
    
    Finally, notice that the following input has $\approx n^2/2$ sensitive bits: for every instance of $\overline{\EH}$ we consider an input with $n/2-1$ ones. Flipping any 0 in this input flips the value of $f$ from 1 to 0. Moreover, it is known that $\deg$ composes and both $\AND$ and $\EH$ have maximal degree, therefore we can conclude that $\deg(f) = \Theta(n^2)$.
\end{proof}

We remark that the function $f=\AND_n \circ \> \EH_n$ also separates $\rdeg(f)$ from $\min\{\C(f), \deg(f)\}$ quadratically. It is clear that $\rdeg(g) \leq \min\{\C(g), \deg(g)\}$ for any function $g$, and $\C(g)$ and $\deg(g)$ can be separated quadratically from each other, so it is not surprising that $\rdeg(g)$ can be quadratically smaller than $\C(g)$ or $\deg(g)$ separately. However, as this example shows, it can be quadratically separated from both simultaneously. 

\section{Applications in Complexity Theory}

In this section, we give two functions: one whose rational degree is unboundedly higher than its approximate degree and one which has approximate degree unboundedly higher than its rational degree. These examples in turn give bidirectional oracle separations between $\bqp$ and $\posteqp$. We conclude the section by giving evidence that zero-error quantum computation with post-selection gives advantage over bounded-error randomized algorithms, providing context to our results.

\subsection{Post-Selection can be a Weak Resource} \label{sec:rdeg-is-high}
In this subsection, we give an oracle which witnesses that  $\bqp \not \subseteq \posteqp$ (in fact, even $\rp \not \subseteq \posteqp$). This is accomplished by constructing a partial function which has constant 1-sided error randomized query complexity but maximal  $\postquantumquery_E$. In fact, this problem also demonstrates that the rational degree can be arbitrarily higher than the approximate degree for partial functions.

\begin{problem}[Majority or None]
    The $\majn_n$ function is defined as a partial Boolean function on the set of bitstrings $x\in \bool^n$ that have Hamming weight either $0$ or at least $n/2$. The function $\majn_n$ takes value $0$ in the former case, and takes value $1$ otherwise.
\end{problem}

\begin{theorem}
\label{thm:majornone-separation} The $\majn_n$ function can be decided by a quantum algorithm using constantly many queries, yet its rational degree is at least $\Omega(n)$.  Consequently, $\majn_n$ witnesses the following separations:
\begin{alignat*}{2}
    \apxdeg(\majn_n) &\leq \bigO(1) \quad \text{ yet }\quad \rdeg(\majn_n)&& \geq \Omega(n), \\
    \quantumquery(\majn_n) &\leq \bigO(1) \quad \text{ yet }\quad \postquantumquery_E(\majn_n)&& \geq \Omega(n).
\end{alignat*}
\end{theorem}

\begin{proof}
    The $\majn_n$ function even has constant $\RP$ query complexity. Indeed, we may simply query a constant number of random bits and output $1$ if any of them are $1$. Therefore, $\majn_n$ has constant quantum query complexity, which in turn implies a constant approximate degree. We show via a rational degree lower bound that $\postquantumquery_E(\majn_n) = \Omega(n)$. In particular, we show that $\rdegnull(\majn_n) = \Omega(n)$. Using the notation of \Cref{lem:rdeg-symmetric}, for $\majn_n$ we have $|S_1| \geq n/2$, giving us
    \[
    \postquantumquery_E(\majn_n)\geq \rdegnull(\majn_n) =
    \Omega(n). \qedhere
    \]
\end{proof}

The complexity classes $\quantumquery$ and $\postquantumquery_E$ are the query complexity equivalents of $\bqp$ and $\posteqp$, respectively. As such, our unbounded separation between these complexity measures gives an oracle separation of $\bqp$ and $\posteqp$.

\begin{corollary} \label{cor:posteqp-low}
    There exists an oracle $O$ such that $\rp^O \not \subseteq \posteqp^O$. \qed
\end{corollary}

\subsection{Post-Selection can be a Strong Resource} \label{sec:rdeg-is-low}

On the other hand, we can give an oracle which witnesses $\posteqp \not \subseteq \bqp$. We do this by constructing a promise problem $f$ for which $\postquantumquery_E(f)=\bigO(1)$ but $\quantumquery_{\eps}(f) = \Omega(n)$. This problem also witnesses the fact that approximate degree can be unboundedly larger than rational degree.

\begin{problem}[\imbalance]
    Let $n = 4m + 2$ for some positive integer $m$. Define the functions $L, R\colon \pmone^n \rightarrow \mathbb{R}$ as $L(x) = x_1 + x_2 + \ldots + x_{2m+1}$ and $R(x) = x_{2m+2} + \ldots + x_{4m+2}$. Then the $\imbalance\colon \pmone^n \to \mathbb{R}$ function is defined as $\imbalance(x) = \frac{L(x)}{R(x)}$. 
\end{problem}

Note that we assumed $4 \nmid n$ to ensure that the denominator $R(x)$ cannot be 0.

\begin{problem}[Boolean Imbalance]
    Let $m$ and $n$ be as in the above problem. We define the Boolean Imbalance $\bi_n$ function as the restriction of $\imbalance$ to the union $S_{-} \cup S_{+}$ where we let
    \begin{align*}
        S_{+} &= \{(x_L, x_R)\colon |x_L| = |x_R| = m\},\\
        S_{-} &= \{(x_L, x_R)\colon |x_L| + |x_R| = 2m+1 \text{ and } |x_L|, |x_R| \geq m\}.
    \end{align*}
\end{problem}

Note that $\bi_n(x) = 1$ for any $x\in S_{+}$ since the numerator and denominator evaluate to the same quantity. On the other hand, for any $x \in S_{-}$ we have that $\bi_n(x) = -1$ since both $L(x)$ and $R(x)$ must be $\pm 1$ but they must be different.

By a generalisation of the equivalence of Mahadev and de Wolf (\cref{lem:rdeg-lowerbound-postq}) we have an upper bound of 2 on $\postquantumquery_E(\bi_n$). We now show that it has a linear lower bound on the approximate degree.

\begin{lemma}
    \label{thm:bi-separation} The $\bi_n$ function can be decided by a post-selected quantum algorithm using only 2 queries, yet its rational degree is at least $\Omega(n)$.  Consequently, $\bi_n$ witnesses the following separations:
\begin{alignat*}{2}
    \rdeg(\bi_n) &\leq \bigO(1) \quad \text{ yet }\quad \apxdeg(\bi_n) &&\geq \Omega(n),
    \\
    \postquantumquery_E(\bi_n) &\leq \bigO(1) \quad \text{ yet }\quad  \quantumquery(\bi_n) &&\geq \Omega(n).
\end{alignat*}
\end{lemma}
\begin{proof}
    Note that $\bi_n$ is defined on inputs of Hamming weight $2m$ and $2m+1$. By a result of Nayak and Wu any function which is constant on Hamming slices $l, l+1$ and flips its value has approximate degree $\Omega(\max\{l, n-l\})$ \cite{Nayak99}. In this case, since the function value flips on Hamming weights $2m, 2m+1$ we get a lower bound of $\Omega(\max\{2m+1, 2m\}) = \Omega(n)$. 
\end{proof}

Finally, just like in the previous section, this separation between complexity measures allows us to construct an oracle relative to which $\posteqp$ is not contained in $\bqp$.

\begin{corollary}\label{cor:posteqp-high}
    There exists an oracle $O$ such that $\posteqp^O \not \subseteq \bqp^O$. \qed   
\end{corollary}

Our unbounded separation of rational degree and approximate degree gives an oracle separation of $\posteqp$ and $\bqp$. Combined with \Cref{cor:posteqp-low}, this tells us that zero-error post-selection and bounded error are \say{incomparable} resources in the black-box model: one is not stronger than the other.

\subsection{Post-Selection and Non-Determinism} \label{sec:post-selection-nondeterminism}
To conclude the section, we provide more context to our results by giving evidence that zero-error quantum computation with post-selection gives advantage over efficient classical computation.
\begin{claim}
\label{claim:np_cap_conp}
$\NP \cap \coNP \subseteq \posteqp$.
\end{claim}
\begin{proof}
    Let $L \in \NP \cap \coNP$. Since $L \in \NP$, there is an efficient algorithm $M_1$ and a polynomial $p_1$ such that for every $x \in L$, there exists $u_1 \in \bool^{p_1(|x|)}$ such that $M_1(x,u_1) = 1$ and for every $x \not \in L$ and $u \in \bool^{p_1(|x|)}$ we have $M_1(x,u) = 0$. 
    Similarly since $L \in \coNP$ there is an efficient algorithm $M_2$ and polynomial $p_2$ such that for every $x \not \in L$, there exists $u_2 \in \bool^{p_2(|x|)}$ such that $M_2(x,u_2) = 1$ and for every $x \in L$ and $u_2 \in \bool^{p_2(|x|)}$ we have $M_2(x,u_2) = 0$. 

    Now, given $x$, our quantum computer can generate a uniform superposition over all the possible certificates for both $M_1$ and $M_2$ (concatenated together), and post-select on the event that either $M(x, u_1) = 1$ or $M_2(x,u_2) = 1$. Then, the quantum algorithm can measure all registers and simulate both $M_1(x,u_1)$ and $M_2(x,u_2)$ and see which one is 1. By definition, only one of $M_1$ and $M_2$ will accept, and whichever one accepts tells us if $x \in L$ or not.
\end{proof}

It is widely believed that $\NP \cap \coNP$ is not contained in $\P$ or even $\bpp$. As such, there is reason to believe that zero-error quantum algorithms with post-selection can offer advantage over efficient classical computation.

\section{Open Questions}

In this paper, we considered the problem of lower bounding the rational degree of Boolean functions in terms of their Fourier degree. While we could not answer this question in its full generality, we showed that the square root of the degree lower bounds the rational degree for both monotone and symmetric Boolean functions. We conjecture that this lower bound extends to all total Boolean functions.
\begin{conjecture}
    For all Boolean functions $f\colon \pmone^n \to \pmone$, $\sqrt{\deg(f)} \leq \rdeg(f)$.
\end{conjecture}

Answering this conjecture in the affirmative would place rational degree within a plethora of Boolean function complexity measures all of which are polynomially related. Recall that for partial functions, we have unbounded separations between the rational and approximate degrees in both directions. 

We showed that a hypothetical total function that witnesses any such separation must lack a certain level of structure: in particular, it cannot be symmetric, monotone, or expressible by a low-depth read-once Boolean formula. In this direction, an easier question is whether there are other classes of functions, such as transitive-symmetric functions, for which rational degree cannot be separated from Fourier degree.

We also proved that almost all Boolean functions $f\colon \pmone^n \to \pmone$ have rational degree at least $n/2 - \bigO(\sqrt{n})$. As mentioned in the preliminaries O'Donnell and Servedio proved \cite{ODONNELL2008298} that almost all Boolean functions $f\colon \pmone^n \to \pmone$ have sign degree at most $n/2 + \bigO(\sqrt{n \log{n}})$. It would be interesting to know if a similar result can be established for the rational degree.

\begin{conjecture}
    All but a negligible fraction of Boolean functions $f\colon \pmone^n \to \pmone$ have rational degree at most $n/2 + \littleO(n)$.     
\end{conjecture}

\subsection*{Acknowledgements}

The authors thank Scott Aaronson, Yuval Filmus, Lance Fortnow, Sabee Grewal, Daniel Liang, Geoffrey Mon, Rocco Servedio, Avishay Tal, Ronald de Wolf, and David Zuckerman for helpful conversations.

VI and SJ are supported by Scott Aaronson’s
Vannevar Bush Fellowship from the US Department of Defense, the Berkeley NSF-QLCI CIQC
Center, a Simons Investigator Award, and the Simons “It from Qubit” collaboration. VI is
 supported by a National Science Foundation Graduate Research Fellowship.
MKD and DW acknowledge support from the Army Research Office (grant W911NF-20-1-0015) and the Department of Energy, Office of Science, Office of Advanced Scientific Computing Research, Accelerated Research in Quantum Computing program. VMK acknowledges support from NSF Grant CCF-2008076 and a Simons Investigator Award (\#409864, David Zuckerman). MW was supported by NSF grant CCF-2006359.

\bibliography{refs}
\bibliographystyle{alphaurl}

\appendix

\newpage
\section{Rational Degree and Query Complexity with Post-Selection}\label{app:rdeg_postq}

In this appendix, we show that the main theorem of \cite{mahadev2014}, stated for total Boolean functions, in fact also holds for partial ones. We do so simply by observing that the proof given in \cite{mahadev2014} also works for partial Boolean functions.

\begin{theorem}
    For any (possibly partial) Boolean function $f$ on $n$ variables and $\eps \in [0,1/2)$, $\rdeg_\eps(f) = \Theta(\postquantumquery_\eps(f))$.
\end{theorem}

First we show that rational degree lower bounds quantum query complexity with post-selection.

\begin{lemma}
\label{lem:rdeg-lowerbound-postq}
    Let $D \subseteq \bool^n$ and consider $f\colon D \to \bool$. Then $\rdeg_\eps(f) \leq \postquantumquery_\eps(f)$.
\end{lemma}

\begin{proof}
    Suppose there is a $T$-query post-selected $\eps$-error algorithm for $f$. As in \cref{def:postbqp} let $a(x)$ be the random variable corresponding to the measurement outcome of the output qubit and $b(x)$ the random variable corresponding to the measurement outcome of the post-selected qubit. We have that 
    \begin{align*}
        \abs{\Pr[a(x) = 1 | b(x) = 1] - f(x)} \leq \eps. 
    \end{align*}
    Now, by \cite{beals1998}, the amplitudes of a quantum algorithm after $T$ oracle queries are polynomials in $x_1,\dots, x_n$ of degree at most $2T$. It immediately follows that
    $\Pr[a(x)\land b(x)=1]$ and $\Pr[b(x)=1]$ are polynomials of degree at most $2T$: call them $p$ and $q$ respectively. Thus we have
    \begin{align*}
        \abs{\dfrac{p(x)}{q(x)} - f(x)} = \abs{\Pr[a(x) = 1 | b(x) = 1] - f(x)} \leq \eps,
    \end{align*}
    which gives us the desired rational approximation of $f$.
\end{proof}

Now, we show that, up to a constant factor, the rational degree upper bounds quantum query complexity with post-selection. The proof is Fourier analytic, and so we switch to the $\pmone$ basis.

\begin{lemma}
\label{lem:postq-lowerbound-rdeg}
    Let $D \subseteq \pmone^n$ and consider $f\colon D \to \pmone$. Then $\postquantumquery_\eps(f) \leq 2 \rdeg_\eps(f)$.
\end{lemma}

\begin{proof}
    Suppose $f\colon D \to \pmone^n$ has an $\eps$-approximate rational representation $p/q$ such that $\max\{\deg(p),\deg(q)\} = d$. Considering the Fourier expansions of $p,q$, we can construct the state
    \[
    \sum_{S \subseteq [n]} \hat{p}(S)\ket{0}\ket{S} + \hat{q}(S)\ket{1}\ket{S},
    \]
    where $\ket{S}$ is the basis state that corresponds to the indicator bitstring for the set $S$. For simplicity, we have left out normalizing constants. Then, using $\max\{\deg(p),\deg(q)\} = d$ queries to $x$, we can construct the state
    \[
    \sum_{S \subseteq [n]} \hat{p}(S) \chi_S(x) \ket{0}\ket{S} + \hat{q}(S) \chi_S(x) \ket{1}\ket{S}.
    \]
    Now we apply an $n$-qubit Hadamard to the second register, obtaining the state
    \[
    \ket{0}\left(\sum_{S \subseteq[n]} \hat{p}(S) \ket{0^n} + ...  \right) + \ket{1}\left(\sum_{S \subseteq[n]} \hat{q}(S) \ket{0^n} + ...  \right)
    \]
    after which we post-select on the second register being equal to $0^n$. This gives us the (again, unnormalized) state
    \[
    \ket{0}\left(\sum_{S \subseteq[n]} \hat{p}(S) \ket{0^n}\right) + \ket{1}\left(\sum_{S \subseteq[n]} \hat{q}(S) \ket{0^n}\right) = \left(p(x)\ket{0} + q(x) \ket{1}\right)\ket{0^n}.
    \]
    After discarding the second register and normalizing, we are left with
    \[
    \frac{p(x)}{p(x)^2 + q(x)^2} \ket{0} + \frac{q(x)}{p(x)^2 + q(x)^2} \ket{1}.
    \]
    We measure this state in the Hadamard basis and interpret the result as having value in $\pmone$. If $f(x) = -1$, then $p(x)/q(x) \in [-1-\eps, -1 + \eps]$. The probability of measuring $\ket{-}$ is 
    \[
    \frac{(q(x) - p(x))^2}{2(p(x)^2 + q(x)^2)} = \frac{(1 - p(x)/q(x))^2}{2((p(x)/q(x))^2 + 1} \leq \frac{\eps^2}{2(1 + (1-\eps)^2)} \leq \eps.
    \]
\end{proof}

Note that in the proof of \cref{lem:rdeg-lowerbound-postq} we get a rational polynomial which is bounded outside of the promise. As a consequence, we have that imposing this boundedness condition only increases the $\epsilon$-approximate rational degree by at most a factor of 2.

\section{Read-Once $\TC$ Formulae}\label{app:readonce_tc}

In this appendix, we prove a structural result on read-once $\TC$ formulae. Recall that $\TC$ formulae are constructed from threshold gates. For convenience in our proof, we will consider the gate set as consisting of the $\NOTgate$ gate and all non-constant \emph{monotone-increasing} Threshold gates. A non-constant monotone-increasing Threshold gate is a gate that computes a threshold function of the form $\Thr_k \colon \{0,1\}^m \to \{0,1\}$ where $m \in \mathbb{N}$, $k\in \{1,\dots, m\}$ and $\Thr_k(x) = 1$ if and only if $\abs{x} \geq k$.

More specifically, we show that any read-once $\TC$ formula on $m$ disjoint variables restricts to either an $\AND$ or $\OR$ function on at least $\sqrt{m}$ distinct variables up to negating some input bits. As such, we make the following definition. 
\begin{definition}[AND- and OR- dimension]
    Let $f\colon \{0,1\}^n \to \{0,1\}$. The AND-dimension of $f$, denoted $\anddim(f)$, is the largest positive integer $m$ such that there exists a restriction $\rho$ of $f$ such that $f|_\rho \colon \{0,1\}^{m} \to \{0,1\}$ is equal to the AND function up to negating some input bits. The OR-dimension of $f$, denoted $\ordim(f)$, is defined similarly using the OR function instead.
\end{definition}

Now we can prove \cref{prop:read_once_threshold}.

\begin{prop}\label{prop:read_once_threshold}
   Suppose $f\colon \{0,1\}^n \to \{0,1\}$ is a read-once $\TC$ formula on $m$ disjoint variables of \emph{any} depth, then $\anddim(f)\cdot \ordim(f) \geq m$. In particular, $\rdeg(f) \geq \max(\anddim(f), \ordim(f)) \geq  \sqrt{m}$.
\end{prop}

\begin{proof}
    The ``in particular'' part, follows from the first part by observing that \cref{lem:rdeg-symmetric,prop:negation} imply $\rdeg(f) \geq \max(\anddim(f), \ordim(f))$. Therefore $\rdeg(f) \geq \sqrt{m}$.

    We proceed to prove the first part. The proof is by induction on the depth $d$ of the tree $T$ defining the read-once formula. 
    
    The proposition holds in the base case $d = 1$ since either 
    \begin{enumerate}
        \item $m=1$ and $f$ is the dictator function or its negation, then $\max(\anddim(f),\ndeg(f)) = 1$, or
        \item $\Thr_k\colon \bool^m \to \bool$ is a restriction of $f$, so $\anddim(f) \geq k$ (by restricting $m-k$ bits of $\Thr_k$ to $0$) and $\ordim(f) \geq m-k+1$ (by restricting $k-1$ bits of $\Thr_k$ to $1$). So in this case $\max(\anddim(f),\ordim(f)) \geq k(m-k+1) \geq m$ as $k\in \{1,\dots,m\}$
    \end{enumerate}

    For the induction step, assume that $T$ has depth $d > 1$, and the proposition holds for all read-once formulae of depth less than $d$. There are again two cases:
    \begin{enumerate}
        \item The root is labeled by the  NOT gate. In this case, $f$ can be  computed by taking the NOT of an read-once formula on $m$ distinct variables of depth $d-1$. By the inductive hypothesis, that read-once formula computes a function $f'\colon \{0,1\}^n \to \{0,1\}$ with $\anddim(f')\cdot \ordim(f') \geq m$. Since $f$ is the NOT of $f'$, De Morgan's laws and the definition of the AND- and OR-dimensions imply that $\anddim(f) \geq \ordim(f')$ and $\ordim(f) \geq \anddim(f')$. Therefore $\anddim(f)\cdot \ordim(f) \geq m$, which establishes the induction step.
        \item The root is labeled by a Threshold gate that computes $\Thr_k \colon \{0,1\}^a \to \{0,1\}$.  Let $T_1, \ldots, T_a$ denote the $a$ read-once formulae feeding into the root. Each $T_i$ is on some number $m_i\geq 1$ of distinct variables and of depth at most $d-1$. $T_j$ and $T_k$ do not share variables for any $j\neq k$ and $\sum_{i=1}^a m_i = m$. Then each $T_i$ can be viewed as computing a non-constant function $f_i\colon \{0,1\}^{m_i} \to \{0,1\}$ such that $ \Thr_k(f_1,\dots, f_a) \colon \{0,1\}^m \to \{0,1\}$ is a restriction of $f$. By the inductive hypothesis, $\anddim(f_i)\cdot \ordim(f_i) \geq m_i$.
        
        Let $A_1 \geq\ldots\geq A_a$ denote the multiset $\{\anddim(f_i): i \in [a]\}$ sorted in descending order. Similarly, let $O_1 \geq\ldots\geq O_a$ denote the multiset $\{\ordim(f_i)\colon i \in [a]\}$ in descending order. Since the $f_i$ functions are all non-constant, we can appropriately restrict $f$ as explained in the base case to deduce
        \begin{equation}\label{eq:induction_and_or_dim}
            \anddim(f) \geq A_1 + \dots + A_k \quad \text{and} \quad \ordim(f) \geq O_1 + \dots + O_{a-k+1}.
        \end{equation}
        We now observe the following fact whose short proof we defer to after this proof. 
        \begin{fact}\label{lem:sequence_inequality}
        Let $n\in \mathbb{N}$. Let $x_1, x_2, \dots, x_n \geq 0$ and $y_1, y_2, \dots,  y_n \geq 0$. Suppose $\tildex_1\geq \tildex_2\geq \dots \geq \tildex_n$ are the $x_i$s sorted in descending order and $\tildey_1 \geq \tildey_2 \geq \dots \geq  \tildey_n \geq 0$ are the $y_i$s sorted in descending order. Then, for any $k\in [n]$,
        \begin{equation}
            (\tildex_1+\tildex_2 + \dots + \tildex_k) (\tildey_1+ \tildey_2 + \dots + \tildey_{n-k+1}) \geq x_1 y_1 + x_2 y_2   + \dots + x_n y_n.
        \end{equation}
        \end{fact}

        Therefore
        \begin{align*}
             \anddim(f) \cdot \ordim(f) \geq& \sum_{i=1}^a \anddim(f_i) \cdot \ordim(f_i) &&\text{(\cref{eq:induction_and_or_dim} 
 and \cref{lem:sequence_inequality})}
             \\
             \geq& \sum_{i=1}^a m_i &&\text{(Inductive hypothesis)}
             \\
             =& m,
        \end{align*}
        which establishes the induction step.
    \end{enumerate}

    The proposition follows.
\end{proof}

\begin{proof}[Proof of \cref{lem:sequence_inequality}]

It suffices to prove that 
\begin{equation}\label{eq:pre_rearrangement}
     (\tildex_1+\tildex_2 + \dots + \tildex_k) (\tildey_1+ \tildey_2 + \dots + \tildey_{n-k+1}) \geq \tildex_1 \tildey_1 + \tildex_2 \tildey_2 + \dots + \tildex_n \tildey_n
\end{equation}
from which the lemma follows by the rearrangement inequality.

To prove \cref{eq:pre_rearrangement}, it suffices to prove it for any $k\leq \ceil{n/2}$ by the symmetry between the $\tildex_i$s and $\tildey_i$s. We may also assume $k \geq 2$ since the case $k=1$ is clearly true.

For $2\leq k \leq \ceil{n/2}$, we have
\begin{align*}
 &(\tildex_1+\tildex_2 + \dots + \tildex_k) (\tildey_1+ \tildey_2 + \dots + \tildey_{n-k+1})
 \\
 \geq& (\tildex_1 + \tildex_2)(\tildey_1+ \tildey_2 + \dots + \tildey_{\floor{n/2}+1}) &&\text{($\tildex_i, \tildey_i\geq 0$)}
 \\
 \geq& \tildex_1 (\tildey_1+ \tildey_2 + \dots + \tildey_{\floor{n/2}+1}) + \tildex_{\floor{n/2} + 2}(\tildey_1+ \tildey_2 + \dots + \tildey_{\floor{n/2}+1})
 \\
 \geq& \tildex_1 \tildey_1 + \tildex_2 \tildey_2   + \dots + \tildex_n \tildey_n,
\end{align*}
as required.     
\end{proof}
\end{document}